\newcommand\MFCSshorter[1]{#1}%

\newcommand\NOTYETSOUMISSIONICALP[1]{}
\renewcommand\NOTYETSOUMISSIONICALP[1]{#1}

\newcommand\coNOTYETSOUMISSIONICALP[1]{#1}
\NOTYETSOUMISSIONICALP{
\renewcommand\coNOTYETSOUMISSIONICALP[1]{}
}

\coNOTYETSOUMISSIONICALP{
}

\newcommand\CORPSPREUVEDE[1]{}

\providecommand\PREUVEPASENANNEXE[1]{#1}
\PREUVEPASENANNEXE{
	\renewcommand\CORPSPREUVEDE[1]{#1}
	
}

\newcommand\SANSCOMMENTAIRE[1]{}
\newcommand\ANONYME[1]{}
\NOTYETSOUMISSIONICALP{
\renewcommand\SANSCOMMENTAIRE[1]{#1}
\renewcommand\ANONYME[1]{#1}
}

\PassOptionsToPackage{unicode}{hyperref}

\newcommand\shortericalp[1]{}
\newcommand\shorter[2][]{}

\newcommand\shortermcu[1]{#1}

\documentclass{jloganal}
\usepackage{pinlabel}

\usepackage{color}
\usepackage{hyperref}

\newtheorem{theorem}{Theorem}[section]
\newtheorem{lemma}[theorem]{Lemma}
\newtheorem{corollary}[theorem]{Corollary}
\newtheorem{proposition}[theorem]{Proposition}

\theoremstyle{definition}
\newtheorem{definition}[theorem]{Definition}
\newtheorem{remark}{Remark}
\numberwithin{equation}{section}
\makeop{Homo}

\usepackage{todonotes}

\newcommand\manondufutur[2][]{\SANSCOMMENTAIRE{\todo[inline,color=pink!80!white,caption={2do},#1]{\begin{minipage}{\textwidth-4pt}
Pour la Manon du Futur:			#2\end{minipage}}}}

\newcommand{\olivierplusimportant}[2][]{}
\newcommand{\manonplusimportant}[2][]{}

\newcommand\spaceclass{\mathbb{RLD}}
\newcommand\spaceclasstanh{\mathbb{RLD}^\circ}
\newcommand\spaceclasstanhlim{\bar{\mathbb{RLD}^\circ}}

\newcommand\Lesfonctionsquoi{\mathcal{F}}
\newcommand\manonclass{\linearderivlength^\bullet}

\newcommand\manonclasslighttanh{\linearderivlength^\circ}
\newcommand\MANONlim{ELim}

\newcommand\manonclasslim{\bar{\manonclass}}

\newcommand\manonclasslighttanhlim{\bar{\manonclasslighttanh}}
\newcommand\polynomialb{ $\bar{\fonction{cond}}$-polynomial}
\newcommand\polynomialbtanh{ $\basictanh$-polynomial}
\newcommand{\basictanh}{\operatorname{\mathfrak{tanh}}}
\newcommand{\signb}[1]{\bar{\fonction{cond}}(#1)}
\newcommand{\signbname}{\bar{\fonction{cond}}}
\newcommand{\signbtanh}[2]{\operatorname{\bar{{\smooth-cond}}}(#1,#2)}

\newcommand\unaire[1]{with respect to the value of $#1$}

\newcommand\EncodeMul{\textit{EncodeMul}}
\newcommand\Decode{\textit{Decode}}
\newcommand\DP{\operatorname{DP}}

\newcommand\send{\operatorname{send}}
\newcommand\smooth{\mathcal{C}}
\newcommand\sendtanh{\operatorname{\smooth-send}}
\newcommand\sendsymbol{\mapsto}
\newcommand\T{\operatorname{\bar {if}}}
\newcommand\TT{\operatorname{if}}
\newcommand\Ttanh{\operatorname{\bar{\smooth-if}}}
\newcommand\TTtanh{\operatorname{\smooth-if}}
\newcommand\round{\operatorname{round}}
\newcommand\relu{\operatorname{ReLU}}

\newcommand\motnouv[1]{\emph{#1}}

\newcommand\THEPAPIERS{\cite{MFCS2019,MFCSJournal}}
\newcommand\THEPAPIERSPLUS{\cite{MFCS2019,MFCSJournal,BlancBournezMCU22vantardise}}
\usepackage{amssymb}
\usepackage{amsmath}
\usepackage{mathdots}
\usepackage[utf8]{inputenc}

\usepackage{tikz}

\usepackage{versions}
\excludeversion{pascettesection}

\usepackage{mathtools}

\newcommand\dyadic{\mathbb{D}}

\newcommand\vectorl[1]{{\mathbf#1}}

\newcommand\vn{\vectorl{n}}

\newcommand{\vertiii}[1]{{\left\vert\kern-0.25ex\left\vert\kern-0.25ex\left\vert #1 
		\right\vert\kern-0.25ex\right\vert\kern-0.25ex\right\vert}}
\newcommand\tnorm[1]{\vertiii{#1}}

\newcommand\N{\mathbb{N}}

 \newcommand\RR{\R}

\newcommand{\fonction}[1]{\textrm{#1}}

\newcommand\projection[2]{\mathbf{\pi}_{#1}^{#2}}

\newcommand\plus{\mathbf{+}}
\newcommand\minus{\mathbf{-}}

\newcommand\dint[4]{\int_{#1}^{#2}{#3}{\delta #4}}

\newcommand\fallingexp[1]{\overline{2}^{#1}}

\newcommand\signname{\fonction{sg}}
\newcommand{\sign}[1]{\fonction{sg}(#1)}

\newcommand{\tu}[1]{\mathbf{#1}}

\newcommand{\cp}[1]{\mathbf{#1}}
\newcommand{\Ptime}{\cp{PTIME}} 

\newcommand{\FPtime}{\cp{FPTIME}}

\newcommand{\FPspace}{\cp{FPSPACE}}

\renewcommand\mod{\operatorname{mod}}
\newcommand\sig{\operatorname{\mathfrak{s}}}
\newcommand\sigtanh{\operatorname{\mathfrak{\smooth-s}}}

\renewcommand\bar[1]{\overline{#1}}

\newcommand{\dderiv}[2]{\frac{\partial #1}{\partial #2}}
\newcommand{\dderivL}[1]{\frac{\partial #1}{\partial \lengt}}
\newcommand{\dderivl}[1]{\frac{\partial #1}{\partial \ell}}

\newcommand{\lengt}{\mathcal{L}}
 \newcommand\lengthnotation{\ell}
\newcommand{\length}[1]{\mathrm{\lengthnotation}(#1)}
\newcommand{\degre}[1]{\mathrm{deg}(#1)}

\newcommand{\linearderivlength}{\mathbb{LDL}}

\newcommand\base{4}
\newcommand\symboleun{1}
\newcommand\symboledeux{3}
\newcommand\encodagemot{\gamma_{word}}

\newcommand\Image{\mathcal{I}}

\newcommand{\Next}{\textit{Next}}
\newcommand{\op}{\textit{op}}
\newcommand{\init}{\textit{init}}
\newcommand{\Formula}{\textit{Formula}}

\title[$\FPtime$ and $\FPspace$ over the reals characterised with discrete ODEs] {Simulation of Turing machines with analytic discrete ODEs:\\  $\FPtime$ and $\FPspace$  over the reals characterised with discrete  ordinary differential equations
}

\author[M. Blanc]{Manon Blanc}
\givenname{Manon}
\surname{Blanc}
\address{Institut Polytechnique de Paris, 
Ecole Polytechnique, LIX, Palaiseau, France}
\email{Manon.Blanc@lix.polytechnique.fr}
\urladdr{https://perso.crans.org/mblanc/}

\author[O. Bournez]{Olivier Bournez}
\givenname{Olivier}
\surname{Bournez}
\address{Institut Polytechnique de Paris, 
Ecole Polytechnique, LIX, Palaiseau, France}
\email{Olivier.Bournez@polytechnique.fr}
\urladdr{https://www.lix.polytechnique.fr/~bournez/}

\keywords{???}
\subject{primary}{msc2000}{03D15,03D20,68Q15,68Q19,39A05,26E05}
\volumenumber{}
\issuenumber{}
\publicationyear{}
\papernumber{}
\startpage{}
\endpage{}
\doi{}
\MR{}
\Zbl{}
\received{}
\revised{}
\accepted{}
\published{}
\publishedonline{}
\proposed{}
\seconded{}
\corresponding{}
\editor{}
\version{}
\begin{document}


%


 \keywords{Discrete ordinary differential equations, Finite Differences, Implicit complexity, Recursion scheme, Ordinary differential equations, Models of computation, Analog Computations}



\begin{abstract}
We prove that functions over the reals computable in polynomial time can be characterised  using discrete ordinary differential equations (ODE), also known as finite differences.   We also provide a characterisation of functions computable in polynomial space  over the reals. 

In particular, this covers space complexity, while existing characterisations were only able to cover time complexity, and were restricted to functions over the integers. We prove furthermore that no artificial sign or test function is needed even for time complexity.

At a technical level, this is obtained by proving that Turing machines can be simulated with analytic discrete ordinary differential equations. We believe  this result opens the way to many applications, as it opens the possibility of programming with  ODEs, with an underlying well-understood time and space complexity. 
\end{abstract}

\maketitle

\section{Introduction}
Recursion schemes constitute a major approach to classical computability theory and, to some extent, to complexity theory. The foundational characterisation of $\FPtime$, based on bounded primitive recursion on notations, due to Cobham \cite{cob65}  gave birth to the field of \textit{implicit complexity} at the interplay of logic and theory of programming. Alternative characterisations, based on safe recursion  \cite{bs:impl} or on 
ramification (\cite{LM93,Lei94}) or for other classes \cite{lm:pspace} followed:  see 
\cite{Clo95,clote2013boolean} for monographs.

Initially motivated to help understand how analogue models of computations compare to classical digital ones, in an orthogonal way, various computability and complexity classes have been recently characterised using Ordinary Differential Equations (ODE).
An unexpected side effect of these proofs is the possibility of programming with classical ODEs, over the continuum. It recently led to solving various open problems. 
This includes the proof of the existence of a universal ODE \cite{ICALP2017}, the proof of the Turing-completeness of chemical reactions \cite{CMSB17vantardise}, or hardness of problems related to dynamical systems \cite{GracaZhongHandbook}. 
\MFCSshorter{
While the above results are easy to state, their proofs are mixing considerations about approximations, control of errors, and various constructions to emulate continuously some discrete processes, despite some recent attempts for a kind of programming languages to help intuition \cite{CIE22}. } 


Discrete ODEs, that we consider in this article,  are an approach in-between born from the attempt of \THEPAPIERS{} to explain some of the constructions for continuous ODEs in an easier way. The basic principle is, for a function $\tu f(x)$, to consider its discrete derivative defined as $\Delta \tu f(x)= \tu f(x+1)-\tu
f(x)$ (also denoted $\tu f^\prime(x)$ in what follows to help analogy with classical continuous counterparts). 
%
A consequence of this attempt is the characterisation obtained in \THEPAPIERS. 
They provided a characterisation of $\FPtime$ for functions  over the integers that does not require the specification of an
explicit bound in the recursion, in contrast to Cobham's work \cite{cob65}, nor the assignment of a specific role or type to variables, in contrast to safe recursion or ramification \cite{bs:impl,Lei-LCC94}. Instead, they only assume involved ODEs to be linear, a very classical natural concept for differential equations.
\MFCSshorter{A characterisation, like ours, happens to be rather simple 
using only common notions from the world of ODEs. In particular, considering \emph{linear} ordinary differential equations is very natural for ODEs.}
\begin{remark}
Unfortunately, even if it was the original motivation, both approaches for characterising complexity classes for continuous and discrete ODEs are currently not directly connected. A key difference is that there is no simple expression (no analogue of the Leibniz rule) for the derivative of the composition of functions in discrete settings. The Leibniz rule is a very basic tool for establishing results over the continuum, using various stability properties, but similar statements cannot be established easily over discrete settings.
\end{remark}

In the context of algebraic classes of functions, the following  notation is classical: call \emph{operation} an operator that takes finitely many functions and returns some new function defined from them. Then $[f_{1}, f_{2}, \dots, f_{k}; \op_{1}, \op_{2},\dots,\op_{\ell}]$ denotes the smallest set of functions containing $f_{1}, f_{2}, \dots, f_{k}$ that is closed under the operations $\op_{1}$, $\op_{2}$, \dots $\op_{\ell}$. 
Call \emph{discrete function} a function of type $ f: S_{1} \times \dots \times S_{d} \to S'_{1} \times \dots S'_{d'}$, where each $S_{i},S'_{i}$ is either $\N$ or $\Z$.
Write $\FPtime$ for the class of functions computable in polynomial time, 
and $\FPspace$ for the class of functions computable in polynomial space.

The literature considers two possible definitions for $\FPspace$, according to whether functions with non-polynomial size values are allowed or not. In our case, we should add ``whose outputs remain of polynomial size'', to resolve the ambiguity. 

\begin{remark}
The point is that, otherwise, the class is not closed by composition. This may be considered as a basic requirement when talking about the complexity of functions.
	The issue is about the usage of not counting the output 
	as part of the total space used. In this model, given $f$ computable in polynomial space, and $g$ in logarithmic space, $f \circ g$ (and $g \circ f$) is computable in polynomial space. But this is not true if we assume only $f$ and $g$ to be computable in polynomial space, since the first might give an output of exponential size.
\end{remark}

A main result of \THEPAPIERS{} is the following ($\linearderivlength$ stands for linear derivation on length):

\begin{theorem}[\cite{MFCS2019}] \label{th:ptime characterization 2}
For  functions over the reals, we have 
	$\linearderivlength= \FPtime$
	where $$\linearderivlength = [\mathbf{0},\mathbf{1},\projection{i}{k}, \length{x}, \plus, \minus, \times, \sign{x} \ ; composition, linear~length~ODE].$$
\end{theorem}
 
In particular, writing as usual $B^{A}$ for functions from $A$ to $B$, we deduce: 
\begin{corollary}[Functions over the integers] 
{$\linearderivlength \cap \N^{\N}= \FPtime \cap \N^{\N}.$}
\end{corollary}

 	That is to say, $\linearderivlength$ (and hence $\FPtime$ for  functions over the integers) is the smallest class of functions
	 that contains the constant functions $\mathbf{0}$ and $\mathbf{1}$, the projections
         $\projection{i}{k}$ of the $i^e$ coordinate of a vector of size $k$,  the length function  $\length{x}$, mapping an integer to the length of its binary representation, 
         the addition $x \plus y$, the subtraction $x \minus y$, the multiplication $x\times y$, 
         the \textbf{sign} function $\sign{x}$ and that is closed under composition (when defined) and linear length-ODE
        scheme: the linear length-ODE scheme, formally given by Definition \ref{def:linear lengt ODE}, corresponds to defining a function from a linear ODE with respect to derivation along the length of the argument,
        so of the form $$\dderivl{\tu f(x,\tu y)} = 	\tu A [\tu f(x,\tu y), 
	x,\tu y]  \cdot \tu f(x,\tu y) 
           + \tu B [\tu f(x,\tu y),
	x,\tu y ].
        $$

 Here, we use the notation $\dderivl{\tu f(x,\tu y)}$ which corresponds to the derivation of $\tu f$ along the length function: given some function $\lengt:\N^{p+1} \rightarrow \Z$ and in particular for the case where $\lengt(x,\tu y)=\ell(x)$,
	\begin{align}\label{lode}
	\dderivL{\tu f(x,\tu y)}= \dderiv{\tu f(x,\tu y)}{\lengt(x,\tu
          y)} = \tu h(\tu f(x,\tu y),x,\tu y)
	\end{align}
	is a formal synonym for
$$ \tu f(x+1,\tu y)= \tu f(x,\tu y) + (\lengt(x+1,\tu y)-\lengt(x,\tu y)) \cdot
\tu h(\tu f(x,\tu y),x,\tu y).$$

\begin{remark}
This concept introduced in \THEPAPIERS, is motivated by the fact that the latter expression is similar to the
classical formula for continuous ODEs: 
$$\frac{\delta f(x,\tu y )}{\delta x} = \frac{\delta
  \lengt (x,\tu y) }{\delta x} \cdot \frac{\delta f(x,\tu
  y)}{\delta \lengt(x, \tu y)},$$
  and hence is similar to a change of variable. Consequently, a linear length-ODE is a linear ODE over a variable $t$ once the change of variable $t=\ell(x)$ is done. 
\end{remark}

However, in the context of (classical) ODEs, considering functions over the reals is more natural than only functions over the integers. 
Call \emph{real function} a function $ f: S_{1} \times \dots \times S_{d} \to S'_{1} \times \dots S'_{d'}$, where each $S_{i},S'_{i}$ is either $\R$, $\N$ or $\Z$.
A natural question about the characterisation of $\FPtime$ for real functions arises, and not only discrete functions: we consider here computability over the reals in its most classical approach, namely computable analysis \cite{Wei00}. 

As a first step, the class 
$$\manonclass= [\mathbf{0},\mathbf{1},\projection{i}{k}, \length{x}, \plus, \minus,\times, \signb{x},\frac{x}{2};\textit{composition, linear length~ODE}]$$ has been considered in \cite{BlancBournezMCU22vantardise,MCUJournalManon} where the authors get some characterisation of $\Ptime$, but only for functions from the integers to the reals (i.e. sequences) while it would be more natural to characterise functions from the reals to the reals. 

More importantly, this was obtained by assuming that some \textbf{non-analytic exact function} is among the basic available functions to simulate a Turing machine: $\signbname$ valuing $1$ for $x>\frac34$ and $0$ for $x<\frac14$. 

We prove first this is not needed, and mainly, we extend all previous results to real functions, furthermore covering not only time complexity but also space complexity.  Consider $$\manonclasslighttanh = [\mathbf{0},\mathbf{1},\projection{i}{k},   \length{x}, \plus, \minus,\tanh,\frac{x}{2},\frac{x}{3};\textit{composition, linear~length~ODE}],$$
where $\ell: \N \to \N$ is the length function, mapping some integer to the length of its binary representation,  $\frac{x}{2}: \R \to \R$ is the function dividing by $2$  (similarly for $\frac{x}{3}$)  and all other basic functions defined exactly as for $\linearderivlength$, but are considered here as functions from the reals to reals.

\begin{remark} This class is $\linearderivlength$ but without the $\sign{x}$ function, nor the multiplication function,  or $\manonclass$ but without the $\signbname$ function, nor the multiplication. This is done by adding the analytic $\tanh$ functions as a substitute (and adding $x/3$).
\end{remark}


\begin{remark}
	We can consider $\N \subset \Z  \subset \R$ but as functions may have different types of outputs, the composition is an issue. We consider, as this is done in \cite{BlancBournezMCU22vantardise,MCUJournalManon}, that
	composition may not be defined in some cases: it is a partial operator. For example, given $f: \N \to \R$ and $g: \R \to \R$, the composition of $g$ and $f$ is defined as expected, but $f$ cannot be composed with a function such as $h: \N \to \N$.
\end{remark}


First, we improve Theorem \ref{th:ptime characterization 2} by stating $\FPtime$ over the integers can be characterised algebraically using linear length ODEs and only analytic functions  (i.e. no need for sign function). 
Since $\manonclasslighttanh$ is about functions over the reals, and Theorem \ref{th:ptime characterization 2} is about functions over the integers, we need a way to compare these classes. Given a function $\tu f: \R^{d} \to \R^{d'}$ sending every integer $\vn \in \N^{d}$ to the vicinity of some integer of  $\N^{d}$, say at a distance less than $1/4$, we write $\DP(f)$ for its discrete part: this is the function from $\N^{d} \to \N^{d'}$ mapping $\vn \in \N^{d}$ to the integer rounding of $\tu f(n)$. Given a class $\mathcal{C}$ of such functions, we write $\DP(\mathcal{C})$ for the class of the discrete parts of the functions of $\mathcal{C}$.

\begin{theorem} \label{trucchosethmain}
$\DP(\manonclasslighttanh)= \FPtime  \cap \N^{\N}.$
\end{theorem}

Second, we improve \cite{BlancBournezMCU22vantardise}: 
 Write $\manonclasslighttanhlim$ for the class obtained by adding some effective limit operation similar to the one considered in \cite{BlancBournezMCU22vantardise} to get $\manonclasslim$. We get a characterisation of functions over the reals (and not only sequences as in \cite{BlancBournezMCU22vantardise}) computable in polynomial time.

\begin{theorem}[Generic functions over the reals] \label{th:main:twop} 
	{$\manonclasslighttanhlim \cap \R^{\R}     = \FPtime \cap \R^{\R} $} \\
	More generally: \quad \quad \quad  \quad  \hspace{0.1cm}
	   {$\manonclasslighttanhlim \cap \R^{\N^{d} \times \R^{d'}} = \FPtime \cap \R^{\N^{d} \times \R^{d}} $.} 
\end{theorem}

We also prove that, by adding a robust linear ODE scheme (Definition \ref{schema:space}), we get a class $\spaceclasstanh$ (this stands for robust linear derivation) with similar statements but for $\FPspace$.
\begin{theorem} \label{trucchosethmainR}
$\DP(\spaceclasstanh)= \FPspace  \cap \N^{\N}.$
\end{theorem}


\begin{theorem}[Generic functions over the reals] \label{th:main:twopSpace} 
	{$\spaceclasstanhlim \cap \R^{\R}  =  \FPspace \cap \R^{\R} $} \\
	More generally: \quad \quad \quad  \quad  \hspace{0.1cm}
	   {$\spaceclasstanhlim \cap \R^{\N^{d} \times \R^{d'}}  = \FPspace \cap \R^{\N^{d} \times \R^{d}} $.} 
\end{theorem}

%
%
%
%

As far as we know, this is the first time a characterisation of $\FPspace$ with discrete ODEs is provided.
If we forget the context of discrete ODEs, $\FPspace$ has been characterised in \cite{thompson1972subrecursiveness} but using a bounded recursion scheme, i.e. requiring some explicit bound in the spirit of Cobham's statement \cite{cob65}. We avoid this issue by considering numerically stable schemes, which are very natural in the context of ODEs. 

At a technical level, all our results are obtained by proving Turing machines can  be suitably simulated with analytic discrete ODEs. 
We believe our constructions could be applied to many other situations, where programming with ODEs is needed.

\paragraph{More on related works} In addition to the previous state-of-the-art discussions, we comment here on some aspects: 
Our ways of simulating Turing machines have some reminiscence of similar constructions used in other contexts such as Neural Networks \cite{SS95,LivreSiegelmann}. In particular, we use a Cantor-like encoding set $\Image$ with inspiration from these references. These references use some particular sigmoid function $\sigma$ (called sometimes the \motnouv{saturated linear function}) that values $0$ when $x \le 0$, $x$ for $0 \le x \le 1$, $1$ for $x \ge 1$. The latter is equivalent to $\signb{\frac14+\frac12x}$, for the function considered in \cite{BlancBournezMCU22vantardise,MCUJournalManon} and hence their constructions can be reformulated using the $\signbname$ function. We completely avoid this, by considering the $\tanh$ function, which is more natural in the context of formal neural networks. 
The models considered in \cite{SS95,LivreSiegelmann}  are recurrent, while our constructions are not recurrent neural networks, and second, their models are restricted to live on the compact domain $[0,1]$, which forbids getting functions from $\R \to \R$, while our settings allow more general functions. Our proofs also require functions taking some integer arguments, that would be impossible to consider in their settings (unless at the price of an artificial recoding). 

\begin{remark} In some sense, our constructions can be seen as operators that map to a family of neural networks in the spirit of these models, instead of considering fixed recurrent neural networks, but also dealing with $\tanh$, and not requiring the saturated linear function. 
\end{remark}

The question of whether Turing machines can be simulated by recurrent neural networks of finite size, using the $\tanh$ activation function (and not the ``exact'' saturated linear function)  is a well-known long-standing open problem. Despite some attempts, such as the one described in \cite{LivreSiegelmann}, up to our knowledge, there remains some of the statements not yet formally proved, or at least not fully generally accepted,  in the existing proofs. Our statements are dealing with the $\tanh$ activation function, in some sense, but we avoid this open question by restricting it to finite space or time computations. By the way, our proofs state this is possible if the space or the time of the machine is bounded, up to some controlled error. 

In the context of neural network models, there have been several characterisations of complexity classes over the discrete (see e.g. the monograph \cite{LivreSiegelmann} about the approach discussed above, but not only), as far as we know, the relation between formal neural networks with classes of computable analysis has never been established before.

If we do not restrict ourselves to neural network-related models, as in all these previous contexts, as far as we know, only a few papers have been devoted to characterisations of complexity, and even computability, classes in the sense of computable analysis. There have been some attempts using continuous ODEs \cite{BCGH07}, that we already mentioned, or 
the so-called $\R$-recursive functions \cite{DBLP:journals/corr/BournezGP16}. For discrete schemata, we only know \cite{brattka1996recursive}  and  \cite{ng2021recursion}, focusing on computability and not complexity.

\paragraph{Organisation of the article} 
In Section \ref{sec:discreteode}, we recall some basic statements about the theory of discrete ODEs. More details can be found in Appendix \ref{sec:dode}, mostly repeated from \THEPAPIERSPLUS, to be self-content.
In Section \ref{sec:functions}, we establish some properties about particular functions required for our proofs. 
 In Section \ref{sec:simulatingmt} we prove our main technical result: Turing machines can be simulated using functions from $\manonclasslighttanh$. Section \ref{sec:convertion} is about converting integers and reals (dyadic) to words of a specific form. 
Section \ref{sec:applications} is about applications of our toolbox. We prove in particular all the above theorems.
%

\paragraph{About appendices} 

In this article, when we say that a function $f: S_{1}  \times \dots \times S_{d} \to \R^{d'}$ is (respectively: polynomial time or space) computable this will always be in the sense of computable analysis: see e.g. \cite{brattka2008tutorial,Wei00}.  As we did not find a reference where the case of functions mixing integer and real arguments is fully formalised, we proposed a formalisation in \cite{ BlancBournezMCU22vantardise}. In order to be self-content, we repeat it in Appendix \ref{sec:defanalysecalculable}. 
In order not to expect our readers to be familiar with the theory of discrete ODEs, as we already wrote, we repeat some basic statements in Appendix \ref{sec:dode}, mostly repeated from \THEPAPIERS. 
As we need some results from \cite{BlancBournezMCU22vantardise}, we also repeat their proof in Appendix \ref{ouestlapreuve}.

\paragraph{Note about current version vs final version} 
This article is the journal version of an article accepted at MFCS'2023 \cite{BlancBournezMFCS2023}. 

\section{Some concepts related to discrete ODEs}
\label{sec:discreteode}
\newcommand\polynomial{ \fonction{sg}-polynomial}
In this section, we recall some concepts and definitions from discrete ODEs, either well-known or established in \THEPAPIERSPLUS: more details, repeated from these references, are provided in Appendix \ref{sec:dode}. 

In order to get a uniform presentation, we consider here that $\basictanh$ is $\tanh$, the hyperbolic tangent. The papers \THEPAPIERS{}  use similar definitions with the sign function $\signname$  and  \cite{BlancBournezMCU22vantardise} with  the piecewise affine function $\signbname$,  which values $1$ for $x>\frac34$ and $0$ for $x<\frac14$, instead of $\tanh$.

\begin{definition}[{\cite{BlancBournezMCU22vantardise}}]
A \polynomialbtanh{} expression $P(x_1,...,x_h)$ is an expression built-on
$+,-,\times$ (often denoted $\cdot$) and $\basictanh{}$ functions over a set of variables $V=\{x_1,...,x_h\}$ and integer constants. 
\end{definition}

We need to measure the degree, similarly to the classical notion of degree in a polynomial expression, but considering all subterms that are within the scope of a $\basictanh$ function contributes to $0$ to the degree. 

\begin{definition}[{\cite{BlancBournezMCU22vantardise}}]
The degree $\deg(x,P)$ of a term  $P$ in $x\in V$  is defined inductively as follows:
\shortermcu{
\begin{itemize}
	\item} $\deg(x,x)=1$ and for  $x'\in V\cup \Z$ such that $x'\neq x$, $\deg(x,x')=0$;
	\shortermcu{\item}  $\deg(x,P+Q)=\max \{\deg(x,P),\deg(x,Q)\}$;
\shortermcu{\item}   $\deg(x,P\times Q)=\deg(x,P)+\deg(x,Q)$;
\shortermcu{\item}   $\deg(x,\basictanh(P))=0$.
\shortermcu{
\end{itemize}}
A \polynomialbtanh{}  expression $P$  is \textit{essentially constant} in
$x$ if $\degre{x,P}=0$. 
\end{definition}

A vectorial function (resp. a matrix or a vector) is said to be a \polynomialbtanh{} expression if all
its coordinates (resp. coefficients) are, 
and 
\textit{essentially constant} if all its coefficients are.

\begin{definition}[\THEPAPIERSPLUS] \label{def:essentiallylinear}
A\polynomialbtanh{} expression $\tu g(\tu f(x, \tu y),  \tu h(x,\tu y), x,
\tu y)$ is \textit{essentially linear} in $\tu f(x, \tu y)$ if
it is of the form: $\tu A [\tu f(x,\tu y),  \tu h(x,\tu y), 
	x,\tu y]  \cdot \tu f(x,\tu y) 
	+ \tu B [\tu f(x,\tu y),  \tu h(x,\tu y), 
	x,\tu y ]$
where $\tu A$ and $\tu B$ are\polynomialbtanh{} expressions essentially
constant in $\tu f(x, \tu y)$. 
\end{definition}

For example, 
\begin{itemize}
\item
    the expression $P(x,y,z)=x\cdot \basictanh{(x^2-z)\cdot y} + y^3$
    is essentially linear in $x$, essentially constant in $z$ and not linear in
    $y$. 
     \item 
      The expression
    $P(x,2^{\length{y}},z)=\signb{x^2 - z}\cdot z^2 + 2^{\length{y}}$
    is essentially constant in $x$, essentially linear in
    $2^{\length{y}}$ (but not essentially constant) and not
    essentially linear in $z$. 
         \item 
      The expression:
    $  z +
    (1-\basictanh(x))\cdot (1-\basictanh(-x))\cdot (y-z) $
    is essentially constant in $x$ and linear in $y$ and $z$.
\end{itemize}

\begin{definition}[Linear length ODE \THEPAPIERS]\label{def:linear lengt ODE}
A function $\tu f$ is linear length-ODE definable from $\tu u$ \textit{essentially linear} in $\tu f(x, \tu y)$, 
$\tu g$ and $\tu h$ if it corresponds to the
solution of 
\begin{equation} \label{SPLode}
f(0,\tu y) 
= \tu g(\tu y)   \quad  and \quad
\dderivl{\tu f(x,\tu y)} 
=   \tu u(\tu f(x,\tu y), \tu h(x,\tu y),
x,\tu y).
\end{equation}
\end{definition}

A fundamental fact is that the derivation with respect to length provides a way to do some change of variables:

\begin{lemma}[{\THEPAPIERS}] \label{myfundobge}
Assume that \eqref{SPLode} holds. 
Then $\tu f(x,\tu y)$ is given by 
$\tu f(x,\tu y)= \tu F(\length{x},\tu y)$
where $\tu F$ is the solution of the initial value problem 
\begin{equation} \label{eq:puissance}
\tu F(1,\tu y)= \tu g(\tu y), \quad and \quad \dderiv{\tu F(t,\tu y)}{t} =\tu u(\tu F(t, \tu y),\tu h(2^{t}-1,\tu y), 2^{t}-1,\tu y).
\end{equation}
\end{lemma}

This means $\tu f(x,\tu y)$ depends only on the length of its first argument:  $\tu f(x,\tu y)= \tu f(2^{\ell(x)},\tu y)$.
 Then \eqref{eq:puissance} can be seen as defining a function (with this latter property) by a recurrence of type
\begin{equation}\label{ourform}
\tu f(2^{0},\tu y)= \tu g(\tu y), \quad and \quad {\tu f(2^{t+1},\tu y)} =\overline{\tu u}(\tu f(2^{t}, \tu y),\tu h(2^{t}-1,\tu y), 2^{t},\tu y).
\end{equation}
for some  $\overline{\tu u}$ is \textit{essentially linear} in $\tu f(2^{t}, \tu y)$. As recurrence \eqref{eq:puissance} is basically equivalent to  \eqref{SPLode}:

\begin{corollary}[Linear length ODE presented with powers of $2$]
A function $\tu f$ is linear $\lengt$-ODE definable iff
the value of $\tu f(x,\tu y)$ depends only on the length of its first argument
and satisfies \eqref{ourform}, 
for some $\tu g$ and $\tu h$, and 
 $\overline{\tu u}$, \textit{essentially linear} in $\tu f(2^{t}, \tu y)$. 
\end{corollary}

We guess it is easier for our reader to deal with recurrences of the form \eqref{ourform} than with ODEs of the form \eqref{SPLode}. Consequently, this is how we will describe many functions from now on, starting with some basic functions, authorising compositions, and the above schemes. As an example, $n \mapsto 2^{n}$ can easily be defined that way. Consider: 
$$\left\{\begin{array}{lll}
2^{0} &=&1 \\
2^{n+1}&=&  2\cdot2^{n}= 2^{n} + 2^{n}
\end{array}\right.$$
Similarly, we can produce $n \mapsto 2^{p(n)}$ for any polynomial $p$. For example,  $$(n_{1},\dots,n_{k}) \to 2^{n_{1} n_{2} \dots n_{k}}$$ can be obtained, using $k$ such schemes in turn, providing the case of the polynomial $p(n)=n^{k}$.

When talking about space complexity, we will also consider the case where the ODE is not derivated with respect to length but with classical derivation. For functions over the reals, an important issue is numerical stability.

\begin{definition}[Robust linear  ODE \THEPAPIERS]\label{def:roblinear lengt ODE} \label{schema:space}
A bounded function $\tu f$ is robustly linear ODE definable from $\tu u$ \textit{essentially linear} in $\tu f(x, \tu y)$, 
$\tu g$ and $\tu h$ if:
\begin{enumerate}
\item it corresponds to the
solution of
\begin{equation} \label{dynamique}
\tu f(0,\tu y) 
= \tu g(\tu y)   \quad  and \quad
\frac{\partial \tu f(x,\tu y)}{\partial x} 
=   \tu u(\tu f(x,\tu y), \tu h(x,\tu y),
x,\tu y),
\end{equation}
\item where  the schema \eqref{dynamique} is polynomially numerically stable.

\end{enumerate}
\end{definition}

Here, writing $a=_{n} b$ for $\| a-b \| \le 2^{-n}$ for conciseness, \textbf{2.} means formally there exists some polynomial $p$ such that, for all integer $n$, writing $\epsilon(n)=p(n+\ell(\tu y))$, if you consider any solution of  
$$\left\{
\begin{array}{lll}
\tilde{\tu y} &=_{\epsilon(n)} & \tu y \\
\tilde{\tu h}(x,\tilde{\tu y}) &=_{\epsilon(n)}& {\tu h}(x,\tilde{\tu y}) \\
 \tilde{\tu f}(0,\tilde{\tu y}) &=_{\epsilon(n)}&
\tu g(\tu y) \\

 \frac{\partial \tilde{\tu f}(x,\tilde{\tu y})}{\partial x} &=_{\epsilon(n)}&
   \tu u(\tilde{\tu f}(x,\tilde{\tu y}), \tilde{\tu h}(x,\tilde{\tu y}),
x,\tilde{\tu y}) 
\end{array}\right.
$$
then
$$\tilde{\tu f}(x,\tilde{\tu y}) =_{\epsilon(n)} \tu f(x,\tu y).$$

\MFCSshorter{This corresponds roughly to the concept of polynomially robust to precision considered in  \cite{BlancBournezLICSsubmitted}, and turns 
out to be a very natural concept for functions defined on a compact.  }
\MFCSshorter{
\begin{remark}
	For linear length ODEs, we did not have to put explicitly numerical stability as a hypothesis, as it comes free from the fact that we consider solutions at most at some logarithmic value of their arguments. But this is required here to guarantee the computability of the solution (and even polynomial space computability). 
	\end{remark}
	\begin{remark}
	Notice that, over the continuum, even computable ODEs may have no computable solution \cite{PouRic79}. Over the discrete, not all dynamics can be simulated, and numerical stability is indeed an issue.
\end{remark}
}
\MFCSshorter{
\begin{remark}
	We believe the hypothesis that $\tu f$ is bounded can be relaxed to: ``the function does not grow as much as the exponential of a polynomial in the length of its arguments'', i.e. not more than a function of $\FPspace$, from arguments similar to the ones of \cite{thompson1972subrecursiveness} about functions over the integers. 
\end{remark} }

\section{Some results about various functions}
\label{sec:functions}

A key part of our proofs is the construction of very specific functions 
in $\manonclasslighttanh$: we write $\{x\}$ for the fractional part of the
real $x$, i.e. $\{x\}=x-\lfloor x \rfloor$. 
%
%

\newcommand\preuvesurtanh{
\begin{lemma}\label{MonLemma4.2.4}
	$\left|1+\tanh x\right|\le 2 \exp(-2|x|) $ for  $x \in(-\infty, 0]$.
\end{lemma}

\newcommand\preuvequatredeuxquatre{}
{
\begin{proof} For $x\le 0$, $|1+\tanh x|=1+\tanh x$, and $|x| = -x$.
	We have $f(x)=2 \exp(2x) - \tanh(x) -1 = 2 \exp(2x) - \frac{1-\exp(-2x)}{1+\exp(-2x)} -1 =\frac{2 \exp(4x)}{1+\exp(2x)} \ge 0$.
%
%
%
%
\end{proof}
}
\CORPSPREUVEDE{\preuvequatredeuxquatre}

\begin{lemma}\label{MonLemma 4.2.5}
	$\left|1-\tanh x\right|\le 2 \exp(-2|x|) $ for  $x \in[0, +\infty)$.
\end{lemma}

\newcommand\preuvequatredeuxcinq{}
{
\begin{proof}
	This follows from Lemma \ref{MonLemma4.2.4}, using the fact that $\tanh$ is odd.
\end{proof}
}
\CORPSPREUVEDE{\preuvequatredeuxcinq}
}

A first observation is that we can uniformly approximate the $\relu(x)=\max(0, x)$ function using an essentially constant function:

\begin{lemma}\label{lemmeRelu}
	Consider (see Figure \ref{fig:relu})  $$Y(x,2^{m+2})= \frac{1+\tanh(2^{m+2} x)}{2}$$
	For all integer $m$, for all $x\in \R$,
	$$| \relu(x) - xY(x, 2^{m+2})| \leq 2^{-m}.$$
\end{lemma}

	\begin{figure}[h]
	\centerline{\includegraphics[height=6cm]{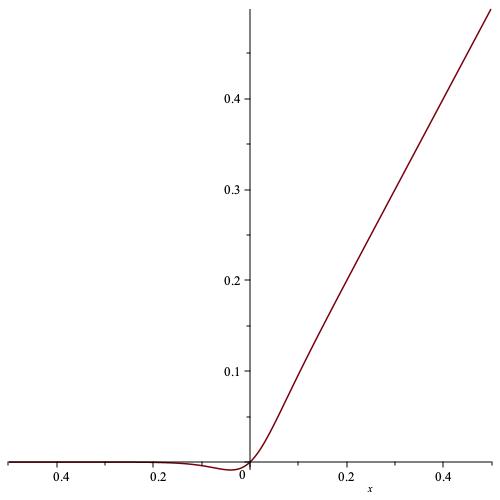}}
	\caption{Graphical representation of $xY(x,2^{2+2})$ obtained with maple.} \label{fig:relu}
	\end{figure}

\newcommand\preuvelemmeRelu{

To prove Lemma \ref{lemmeRelu}, we start by the following basic facts about function $\tanh$.

\preuvesurtanh

\begin{proof}[Proof of Lemma \ref{lemmeRelu}]
	Let $m \in \N$. Consider $Y(x,K)= \frac{1+\tanh(K x)}{2}$, with $K >0$.
	
	For $0 \le x$, $\relu(x)=x$, and $|\relu(x) - xY(x, K)|= \frac{x}{2}|1-\tanh(Kx)| \le x \exp(-2Kx)$
	from Lemma \ref{MonLemma 4.2.5}.
	
	For $x \le 0$, $\relu(x)=0$, and $|\relu(x) - xY(x, K)|= \frac{|x|}{2} |1+\tanh(K x)| \le |x| \exp(-2K|x|)$ 
	from Lemma \ref{MonLemma4.2.4}, which is the same expression as above for $0 \le x$.
	
	Function $g(x)=x \exp(-2Kx)$ has its derivative $g'(x)= \exp(-2Kx) (1-2Kx)$. We deduce the maximum
	of this function $g(x)$ over $\R$ is in $\frac{1}{2K}$, and that the maximum value of $g(x)$ is $\frac{1}{e2K}$.
		
	Consequently, if we take $K=2^{m+2}$, then $g(x) \le 2^{-m}$ for all $x$, and we conclude.
\end{proof}


}
\CORPSPREUVEDE{\preuvelemmeRelu}

We deduce we can uniformly approximate the continuous sigmoid functions (when $1/(b-a)$ is in $\manonclasslighttanh$) defined as: $\sig(a,b,x) = 0$ whenever $w \leq a$, $\frac{x-a}{b-a}$ whenever $a \le x \le b$, and $1$ whenever $b \leq x$. 

\begin{lemma}[Uniform approximation of any piecewise continuous sigmoid] \label{lem:solution:a:tout}
Assume $a,b,\frac1{b-a}$ is in $\manonclasslighttanh$. Then there is some function (illustrated by Figure \ref{fig:pasrelu}) $\sigtanh(z,a,b,x) \in \manonclasslighttanh$ such that for all integer $m$, $$|\sigtanh(2^{m},a,b,x) - \sig(a,b,x)|\le 2^{-m}.$$
\end{lemma}

	\begin{figure}[h]
	\hfill \includegraphics[width=6cm]{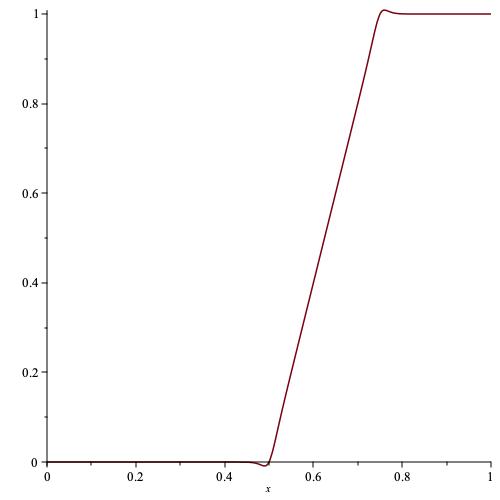} \hfill \includegraphics[width=6cm]{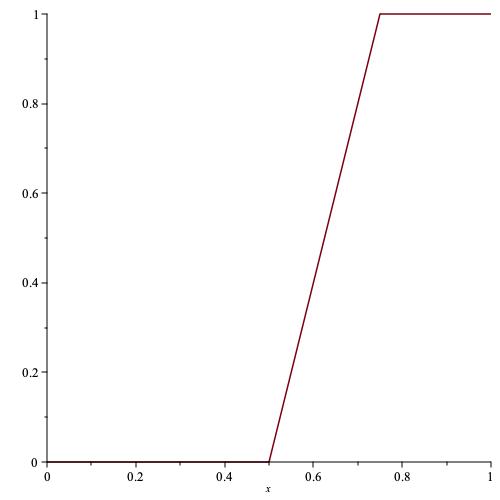} \hfill 
	\caption{Graphical representation of $\sigtanh(2,\frac12,\frac34,x)$ and $\sigtanh(2^{5},\frac12,\frac34,x)$ obtained with maple.} \label{fig:pasrelu}
	\end{figure}

\begin{proof}
	We can write $\sig(a,b,x) = \frac{\relu(x-a)-\relu(x-b)}{b-a}$.  Consider $\sigtanh(z, a, b, x)=  \frac{(x-a) Y(x-a, z2^{1+c}) - (x-b) Y(x-b, z2^{1+c}) }{b-a}$. Thus, 
	$| \sigtanh(2^{m+1+c}, a, b, x) - \sig(a,b,x)| \leq 
	\frac{2. 2^{-m-1-c}}{b-a}$, using the triangle inequality. Take  $c$ such that $\frac1{b-a} \le 2^{c}$.
\end{proof}

The existence of the following function will play an important role to obtain some various other functions.

\begin{theorem} \label{th:xifonda}
 There exists some function (illustrated by  Figure \ref{fig:xideuxun}) $\xi: \N^{2}\to \R$ in $\manonclasslighttanh$ such that for all $n,m\in \N$ and $x\in [- 2^{n} , 2^{n}]$, whenever $ x \in [\lfloor x \rfloor + \frac{1}{8}, \lfloor x \rfloor + \frac{7}{8}],$ $$\left|\xi(2^m,2^{n},x)  - \left\{ x -\frac18 \right\}\right| \le 2^{-m}.$$
 \end{theorem}
 
 \begin{figure}[h]
	\hfill \includegraphics[width=6cm]{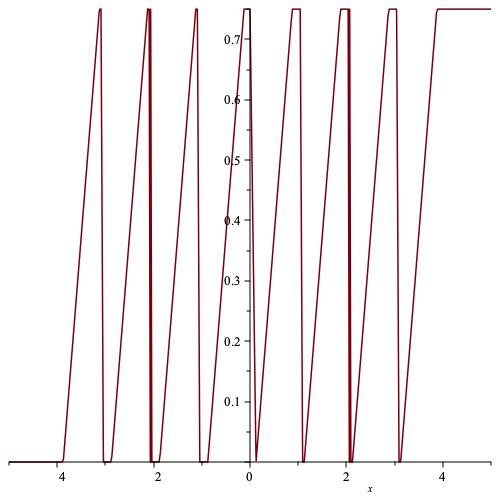} 
	\hfill
	\includegraphics[width=6cm]{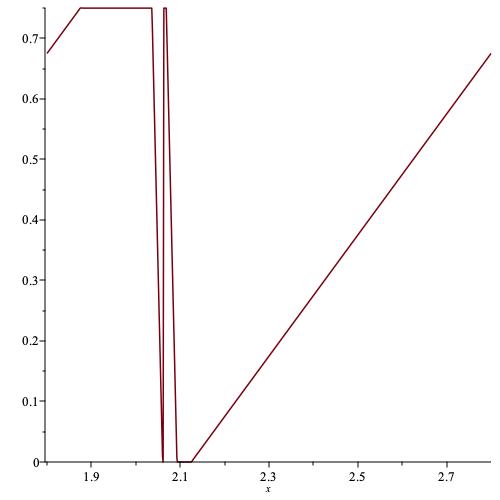}
	\hfill 
	\caption{Graphical representations of $\xi(2,4,x)$  obtained with maple: some details on the right. \label{fig:xideuxun}}
\end{figure}

%
%

\newcommand\PREUVEXI{
\begin{proof}
	If we take $\xi'$ that satisfies the constraint only when $x \ge 0$, and that values $0$ for $x \le 0$, then $\frac34-\xi'(\cdot,\cdot,-x)$ would satisfy the constraint when $x \le 0$, but values $3/4$ for $x \ge 0$. So,  $$\xi(2^m, N,x)=\xi'(2^{m+2},N,x)-\xi'(2^{m+2},N,-x)+\frac34-\frac34 \sigtanh(2^{m+2},0,\frac18,x)$$ would work for all $x$. 
	So it remains to construct $\xi'$ such that for all $n\in \N$, $x\in [0 , 
	2^{n}]$ and $m\in\N$, whenever $ x \in [\lfloor x \rfloor + \frac{1}{8}, 
	\lfloor x \rfloor + \frac{7}{8}] $  , $|\xi'(2^{m},2^n,x)-\{ 
	x-\frac18\}| \le 2^{-m}$, and $|\xi'(2^{m},N,x)-0] \le 2^{-m}$ for $x \le 0$.
	
	Let $s(x)=\frac34\sig(\lfloor x \rfloor +\frac18,\lfloor x \rfloor +\frac78,x)$. Over $[\lfloor x \rfloor + \frac{1}{8}, 
	\lfloor x \rfloor + \frac{7}{8}]$, we have $s(x)=\{ 
	x-\frac18\}$. Actually, we will even construct $\xi'$ with the stronger properties that   whenever $ x \in [\lfloor x \rfloor + \frac{1}{8} -2^{-m}, 
	\lfloor x \rfloor + \frac{7}{8}+2^{-m}]$, $|\xi'(2^{m},2^n,x)-s(
	x-\frac18)| \le 2^{-m}$.
	
	It suffices to define $\xi'$  by induction by $$
	\left\{ \begin{array}{lll} \xi'(2^{m},2^{0},x) 
	&=& \frac34 	\sigtanh(2^{m},\frac18,\frac78,x) \\ 
	\xi'(2^{m},2^{n+1},x)&=& \xi'(2^{m+1},2^{n},F(2^{m+1},2^{n},x))
	\end{array}\right.$$
	where $$F(2^{m+1},K,x)= x- K.\sigtanh(2^{m+1},K+\frac1{32},K+\frac3{32},x).$$
	
	Let $I_{\lfloor x \rfloor}$ be $[\lfloor x \rfloor + \frac{1}{8}, \lfloor x \rfloor + \frac{7}{8}] $,   $x \in  I_{\lfloor x \rfloor}$, and let us first study the value of $F(2^{m+1},2^{n},x)$: 
	
	\begin{itemize}
		\item If $x \leq 2^{n}$, by definition of $\sigtanh$, $|F(2^{m+1},2^n,x) - x| \le 2^{-(m+1)}$, with $x \in  I_{\lfloor x \rfloor}$.
		\shortericalp{	\item The case $2^{n} < x < 2^{n} + \frac18$ cannot happen as we assume $x \in  I_{\lfloor x \rfloor}$.}
		\item If $2^{n} + \frac18  \leq x $ then $|F(2^{m+1},2^n, x) - (x-2^n)| \le 2^{-(m+1)}$ with $x-2^n \in  I_{ \lfloor x \rfloor-2^{n}} $.	\end{itemize}
	
	Now, the property is true by induction. Indeed, it is true for $n=0$ by definition of $\xi'(2^{m},2^{0},x)$. We now assume it is true for some $n\in \N$. We have
	$\xi'(2^{m},2^{n+1},x)=\xi'(2^{m+1},2^{n},F(2^{m+1},2^{n},x))$. Thus, by induction hypothesis, $|\xi'(2^{m+1},2^{n},F(2^{m+1},2^{n},x))-s( F(2^{m+1},2^{n},x) -1/8)| \le 2^{-(m+1)}$.
	Now:
	\begin{itemize}
		\item If $x \leq 2^{n}$, by definition of $\sigtanh$, $|F(2^{m+1},2^n,x) - x| \le 2^{-(m+1)}$, and as $s$ is $1$-Lipschitz,
		$|s( F(2^{m+1},2^{n},x) -\frac18) -s(x-\frac18) | \le |F(2^{m+1},2^{n},x) - x | \le 2^{-(m+1)}$. Consequently,
		$|\xi'(2^{m},2^{n+1},x)-s(
		x-\frac18)| \le 2^{-m}$ and the property holds for $n+1$.
		
		\shortericalp{	\item The case $2^{n} < x < 2^{n} + \frac18$ cannot happen with our constraint $x \in  I_{\lfloor x \rfloor}$.}
		\item If $2^{n} + \frac18  \leq x $ then $|F(2^{m+1},2^n, x) -(x-2^n)| \le 2^{-(m+1)} $  and as $s$ is $1$-Lipschitz,
		$|s( F(2^{m+1},2^{n},x) -\frac18) -s(x-2^n- \frac18) | \le |F(2^{m+1},2^{n},x) - x-2^n | \le 2^{-(m+1)}$. Consequently,
		$|\xi'(2^{m},2^{n+1},x)-s(
		x-2^n-\frac18)| \le 2^{-m}$ and the property holds for $n+1$.
	\end{itemize}
	
	There remains to prove that the function $\xi'$ is in $\manonclasslighttanh$. Unfortunately, this is not clear from the recursive definition, but this can be written in another way, from which this follows. Indeed, we have from an easy induction that  $\xi'(2^{m},2^{n},x)= F(2^{m+n-1},2^{0},F(2^{m+n-2},2^{1},F(2^{m+n-3},2^{2}$ $(\dots, F(2^{m},2^{n-1},x))))),$ if we define $$F(2^{m},2^{0},x)=\xi'(2^{m},2^{0},x) = \frac34 \sigtanh(2^{m},\frac18,\frac78,x).$$
	
	Then, we can obtain $\xi'(2^{m},2^{n},x)=H(2^{m},2^{n-1},2^{n},x)$ with 
	\begin{align*}
		H(2^{m},2^{0}, 2^n, x) &= F(2^{m},2^{n-1}, x) \\
		H(2^{m},2^{t+1}, 2^n, x) &= F(2^{m+t},2^{n-1-t}, H(2^{m},2^{t}, 2^n, x)) \\
		&= H(2^{m},2^{t}, 2^n, x) - 2^{n-1-t}.\sigtanh(2^{m+t},2^{n-1-t},2^{n-1-t}\\
		& \quad +\frac18, H(2^{m},2^{t}, 2^n, x))
	\end{align*}
	
	Such a recurrence can be then seen as a linear length ordinary differential equation, in the length of its first argument. It follows that $\xi'$ is in 
	$\manonclasslighttanh$. 
\end{proof}
}

\CORPSPREUVEDE{\PREUVEXI}

\MFCSshorter{
From the construction of the previous  functions, we obtain a bestiary of various functions
\olivierplusimportant{vérifier effets de bords}}

\newcommand\statementun{$\xi_1, \xi_2 : \N^2 \times \R \mapsto \R ~  \in \manonclasslighttanh$  
	such that, for all $n,m \in \N$,  $\lfloor x \rfloor \in [- 2^{n}+1, 2^{n}]$, 
	\begin{itemize}
	\item 
	 whenever $ x \in [\lfloor x \rfloor - \frac{1}{2}, \lfloor x \rfloor + \frac{1}{4}], $  $$|\xi_1(2^m, 2^n,x) - \{ x \}| \le 2^{-m},$$
	 \item 
	 and whenever $ x \in [\lfloor x \rfloor, \lfloor x \rfloor + \frac{3}{4}],$  $$|\xi_2(2^m, 2^n,x)-\{ x \}| \le 2^{-m}.$$
	 \end{itemize}
	 }

\newcommand\statementdeux{
$\sigma_1, \sigma_2 : \N^2 \times \R \mapsto \R ~  \in \manonclasslighttanh$  
	such that, for all $n,m \in \N$,  $\lfloor x \rfloor \in [-2^{n}+1 , 2^{n}]$,
	\begin{itemize}
	\item 
	 whenever  $ x \in [\lfloor x \rfloor - \frac{1}{2}, \lfloor x \rfloor + \frac{1}{4}]$, $$|\sigma_1(2^m, 2^n,x)-\lfloor x \rfloor|\le 2^{-m},$$
	 \item
	 and whenever  $ x \in I_{2}=[\lfloor x \rfloor, \lfloor x \rfloor + \frac{3}{4}]$, $$|\sigma_2(2^m, 2^n,x)-\lfloor x \rfloor|\le 2^{-m}.$$
	 \end{itemize}
}

\newcommand\statementtrois{
$\lambda : \N^2 \times\R \mapsto [0,1]   \in \manonclasslighttanh$  
 such that for all $m,n\in\N$,  $\lfloor x \rfloor \in [-2^{n}+1 , 2^{n}]$, 
 \begin{itemize}
 \item
 whenever  $ x \in [\lfloor x \rfloor + \frac{1}{4}, \lfloor x \rfloor + \frac{1}{2}] $, $$|\lambda(2^m, 2^n,x)-0| \le 2^{-m},$$
 \item 
	 and whenever $  x \in [\lfloor x \rfloor + \frac{3}{4}, \lfloor x \rfloor +1] $, $$|\lambda(2^m, 2^n,x)-1|\le 2^{-m}.$$
	 \end{itemize}
}

\newcommand\statementquatre{
$\mod_{2} : \N^2 \times\R \mapsto [0,1]   \in \manonclasslighttanh$  
such that for all $m,n\in\N$,  $\lfloor x \rfloor \in [-2^{n}+1 , 2^{n}]$, 
	whenever  $ x \in [\lfloor x \rfloor -\frac14, \lfloor x \rfloor
	 + \frac{1}{4}] $, $$|\mod_{2}(2^m, 2^n, x)-\lfloor x \rfloor \mod 2| \le 2^{-m}.$$
}

\newcommand\statementcinq{$\div_{2} : \N^2 \times\R \mapsto [0,1]  \in \manonclasslighttanh$  
	 such that for all $m,n\in\N$,  $\lfloor x \rfloor \in [-2^{n}+1 , 2^{n}]$, 
	 whenever  $ x \in [\lfloor x \rfloor - \frac14, \lfloor x \rfloor +
	  \frac{1}{4}] $, $$|\div_{2}(2^m, 2^n, x)-\lfloor x \rfloor / / 2| \le 2^{-m},$$ where  $/ /$ is  the integer division.}


%
%
%

\newcommand\preuvebestiary{


\begin{corollary}\label{lem:xi}
	There exists (see Figure \ref{fig:xi}) \statementun
\end{corollary}
\begin{figure}
	\hfill 
	\includegraphics[width=6cm]{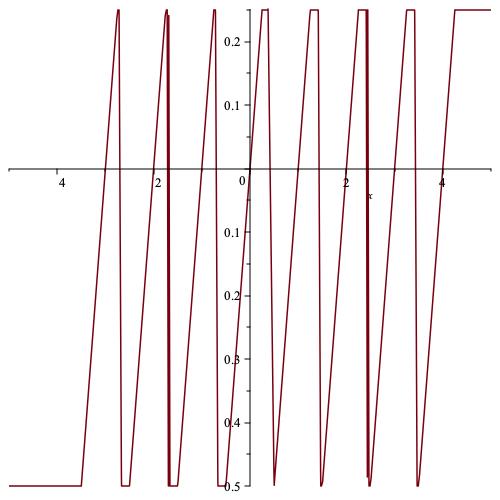} \hfill 
	\includegraphics[width=6cm]{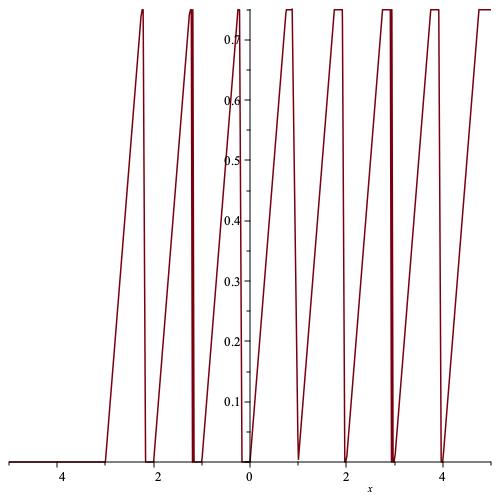}
	\hfill 
	\caption{Graphical representation of $\xi_{1}(2,4,x)$ and $\xi_{2}(2,4,x)$ obtained with maple. \label{fig:xi}}
\end{figure}

\begin{proof}
	Consider $$\xi_{1}(2^m, N,x)=\xi(2^m, N,x-\frac38) -\frac12$$ 
	and $$\xi_{2}(2^m, N,x)=\xi(2^m, N,x-\frac78).$$
\end{proof}

\begin{corollary}\label{lem:i}
There exists (see Figure \ref{fig:sigma}) \statementdeux 
\end{corollary}

\newcommand{\preuvelemi}{}
{
\begin{proof}
	Consider $\sigma_i (2^n,x) = x - \xi_i(2^n,x) $ with the function defined in Corollary \ref{lem:xi}.
\end{proof}

\begin{figure}
	\hfill 
	\includegraphics[width=6cm]{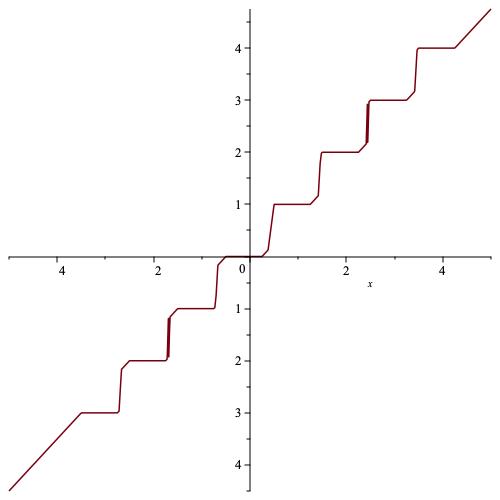} \hfill 
	\includegraphics[width=6cm]{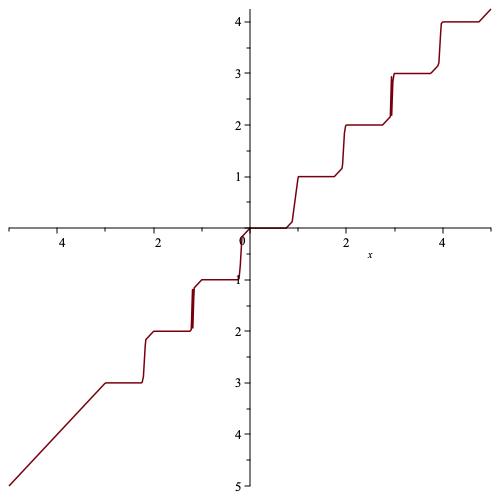}
	\hfill 
	\caption{Graphical representation of $\sigma_{1}(2,4,x)$ and $\sigma_{2}(2,4,x)$ obtained with maple. \label{fig:sigma}}
\end{figure}
}
\CORPSPREUVEDE{\preuvelemi}

\begin{corollary}\label{lem:lambda}
	There exists (see Figure \ref{fig:lambda})  \statementtrois
\end{corollary}


\begin{proof}
Consider $\lambda(2^m, 2^{n},x)= F(\xi(2^{m+1}, 2^{n},x-9/8))$ where $$F(x)=\sigtanh(2^{m+1},1/4,1/2,x).$$
\end{proof}


\begin{figure}
	\centerline{
		\includegraphics[width=6cm]{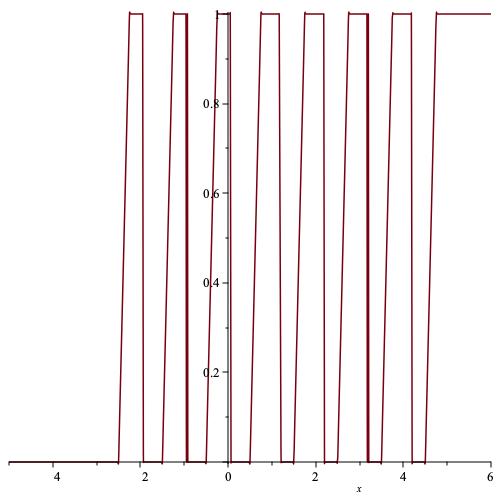} 
	}
	\caption{Graphical representation of $\lambda(2,4,x)$  obtained with maple. \label{fig:lambda}}
\end{figure}


\begin{corollary}\label{lem:mod2}
	There exists (see Figure \ref{fig:moddeux})  \statementquatre
	\end{corollary}

\newcommand\preuvelemmoddeux{}
{
\begin{proof}
	We can take $$\mod_{2}(2^m, N,x)=1-\lambda(2^m, N/2,\frac12x+\frac78)$$ where $\lambda$ is the function given by Corollary \ref{lem:lambda}.
\end{proof}

\begin{figure}
	\centerline{
		\includegraphics[width=6cm]{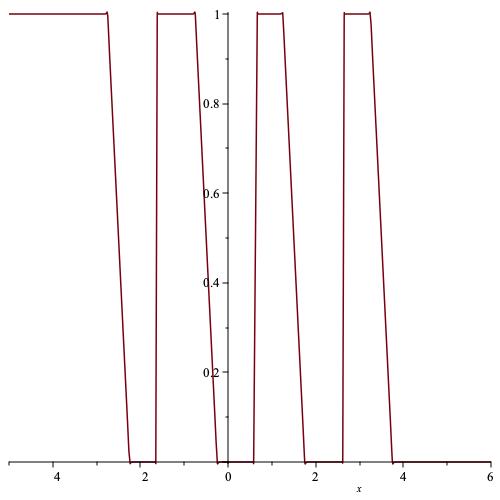} 
	}
	\caption{Graphical representation of $\mod_{2}(2,4,x)$  obtained with maple. \label{fig:moddeux}}
\end{figure}
}

\CORPSPREUVEDE{\preuvelemmoddeux}

\begin{corollary}\label{lem:div2}
	There exists (see Figure \ref{fig:divdeux}) \statementcinq
\end{corollary}

\newcommand\preuvelemdivdeux{}
{
\begin{proof}
	We can take $$\div{2}(2^m,N,x)=\frac12(\sigma_{1}(2^m,N,x)-\mod_{2}(2^m, N,x))$$
	where $\mod_{2}$ is the function given by Corollary \ref{lem:mod2}, 
	and $\sigma_{2}$ is the function given by Corollary \ref{lem:i}.
	
\end{proof}

\begin{figure}
	\centerline{
		\includegraphics[width=6cm]{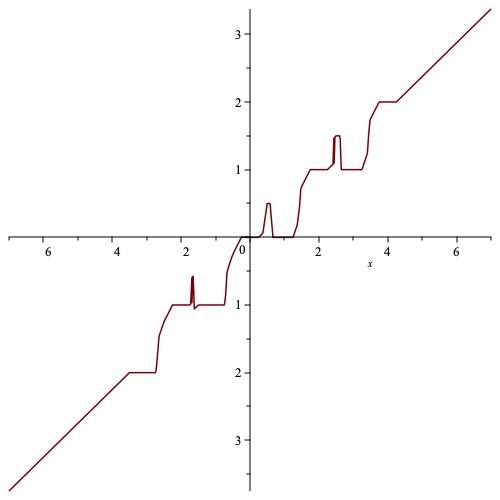} 
	}
	\caption{Graphical representation of $\div_{2}(2,4,x)$  obtained with maple. \label{fig:divdeux}}
\end{figure}
}

\CORPSPREUVEDE{\preuvelemdivdeux}

}

\CORPSPREUVEDE{\preuvebestiary}

\shortermcu{
\subsection{Some properties of sigmoids}
}

%
Observing that for function $$\T(d,\ell) = 4 \sig(1,2,1/2+d+\ell/4) - 2,$$ for $\ell \in [0,1]$, we have 
$$\begin{array}{lll}
\T(0,\ell)&=&0 \\
\T(1,\ell)&=&\ell
\end{array}
.$$
 Applying Lemma \ref{lem:solution:a:tout} on this sigmoid, we get:
\begin{lemma} \label{tricksigmoidtanh}
There exists $\TTtanh \in \manonclasslighttanh$ such that, for $\ell \in [0,1]$, 
\begin{itemize}
\item if we take $|d'-0| \le 1/4$, 
then $|\TTtanh(2^{m},d',\ell)-0| \le 2^{-m},$ 
\item and if we take $|d'-1| \le 1/4$, then $|\TTtanh(2^{m},d',\ell)-\ell| \le 2^{-m}$.
\end{itemize}
\end{lemma}

\newcommand\preuvetricksigmoidtanh{

We start by proving the following lemma, which is about some ideal sigmoids:

\begin{lemma} \label{tricksigmoid}
For $\ell \in [0,1]$, for $-1/4 \le d \le 1/4$, $\T(d,\ell)=0$, and  for $3/4 \le d \le 5/4$, $\T(d,\ell)= 4 (d-1) + \ell$.

Consider $\TT(d,\ell)= \T(\sig(1/4,3/4,x),\ell)$. 
\begin{itemize}
\item If we take $|d'-0| \le 1/4$, 
then $\TT(d',\ell)=0$, 
\item and if we take $|d'-1| \le 1/4$, then $\TT(d',\ell)=\ell$.
\end{itemize}
\end{lemma}

\begin{proof}
	Just check that for $d \le 1/4$, we have $1/2+d+ \ell/4 \le 1$, and hence $\sig(1,2,1/2+d+\ell/4)=0$, so $\T(d,\ell)=0$.  
	For $3/4 \le d \le 5/4$, we have $5/4 \le 1/2+d+\ell/4 \le 2$, and hence $\sig(1,2,1/2+d+\ell/4)=d+\ell/4-1/2$, and hence $\T(1,\ell)=\ell$. The other observation follows.
\end{proof}


Then we go to versions using $\tanh$:

\begin{lemma} 
Consider $\Ttanh(2^{m},d,l) = 4 \sigtanh(2^{m+2},1,2,1/2+d+l/4) - 2.$
For $\ell \in [0,1]$, we have $|\Ttanh(0,\ell)-0| \le 2^{-m},$  and $|\Ttanh(1,\ell)-\ell| \le 2^{-m}.$

For $-1/4 \le d \le 1/4$, $|\Ttanh(2^{m},d,\ell)-0| \le 2^{-m},$  and for $3/4 \le d \le 5/4$, $|\Ttanh(2^{m},d,\ell) -4 (d-1) + \ell| \le 2^{-m}$.

Consider $\TTtanh(2^{m},d,\ell)= \Ttanh(2^{m+1},\sigtanh(2^{m+1},1/4,3/4,x),\ell)$.
\begin{itemize}
\item  If we take $|d'-0| \le 1/4$, 
then $|\TTtanh(d',\ell)-0| \le 2^{-m}$, 
\item and if we take $|d'-1| \le 1/4$, then $|\TTtanh(d',\ell)-\ell| \le 2^{-m}$.
\end{itemize}
\end{lemma}

\begin{proof}
This is direct from previous lemma and Lemma \ref{lem:solution:a:tout}.
\end{proof}

Then Lemma \ref{tricksigmoidtanh} follows.
}
\CORPSPREUVEDE{\preuvetricksigmoidtanh}


\olivierplusimportant{En fait, ca peut donc se faire avec $\sigtanh$ sans trop de pb.}
\olivierplusimportant{
\manondufutur{
Etait écrit: 
Maisc ompliqué par rapport à la solution du dessus, sauf si on s'interdit de mutiplier des $x$ par des $\tanh(x)$.

If we want to do something similar using $\tanh$, from above descriptions, $$T(d,l)=100 \tanh(l/100 + (d - 1) 100) + 100 \tanh(-d + 1)/\tanh(1)$$
is very close to $0$ for $d=0$, and very close to $l$ for $d=1$, when $l \in [-1..1]$.  But, NOTICE THAT THIS expression requires the inverse of $\tanh(1)$.
}}

\begin{lemma}\label{lem:switch:tanh}
Let $\alpha_{1},  \alpha_{2},\dots,$ $\alpha_{n}$ be some integers, and $V_{1}, V_{2}, \dots,V_{n}$ some constants. There is some  function in $\manonclasslighttanh$, that we write 
$\sendtanh(2^{m},{\alpha_{i} \sendsymbol V_{i}})_{i \in \{1,\dots,n\}}$, that maps any $x \in [\alpha_{i}-1/4, \alpha_{i}+1/4]$ to a real at a distance at most $2^{-m}$ of $V_{i}$, for all $i \in \{1,\dots,n\}$. 
\end{lemma}

\newcommand\preuvesend{

We prove the following lemma about some ideal sigmoids: 
define $\signb{x}$ as $\sig(1/4,3/4,x)$ and 
 $\signbtanh{2^{m}}{x}$ as $\sigtanh(2^{m},1/4,3/4,x)$.

\begin{lemma}\label{lem:switch}
Assume you are given some integers $\alpha_{1},  \alpha_{2},\dots,$ $\alpha_{n}$, and  some constants  $V_{1}, V_{2}, \dots,V_{n}$. Then there is some  function in $\manonclasslighttanh$, written $\send({\alpha_{i} \sendsymbol V_{i}})_{i \in \{1,\dots,n\}}$, that maps any $x \in [\alpha_{i}-1/4, \alpha_{i}+1/4]$ to $V_{i}$, for all $i \in \{1,\dots,n\}$. 
\end{lemma}

\begin{proof}
Sort the $\alpha_{i}$ so that $\alpha_{1} < \alpha_{2} <\dots, \alpha_{n}$. Then consider $T_{1} + \signb{x-\alpha_{1}} (T_{2}- T_{1}) + \signb{x-\alpha_{2}} (T_{3}- E_{3}) + \dots + \signb{x-\alpha_{n-1}} (T_{n}- T_{n-1}).$
\end{proof}

We now go to versions with $\tanh$.


\begin{proof}[Proof of Lemma \ref{lem:switch:tanh}]
Sort the $\alpha_{i}$ so that $\alpha_{1} < \alpha_{2} <\dots, \alpha_{n}$. Then consider $T_{1} + \signbtanh{2^{m+c}}{x-\alpha_{1}} (T_{2}- T_{1}) + \signbtanh{2^{m+c}}{x-\alpha_{2}} (T_{3}- T_{2}) + \dots + \signbtanh{2^{m+c}}{x-\alpha_{n-1}} (T_{n}- T_{n-1}).$
for some constant $c$ so that $n \max_{j} (T_{j}-T_{j+1}) \le 2^{c}$.
\end{proof}

}
\CORPSPREUVEDE{\preuvesend}

{More generally:}

\begin{lemma} \label{lem:switch:pairs:tanh}
Let $N$ be some integer. Let $\alpha_{1},  \alpha_{2},\dots, \alpha_{n}$ be some integers, and   $V_{i,j}$ for $1 \le i \le n$ some constants, with $0 \le j < N$. Then there is some  function in $\manonclasslighttanh$, that we write 
$\sendtanh(2^{m},{(\alpha_{i},j) \sendsymbol V_{i,j}})_{i \in \{1,\dots,n\}, j \in \{0,\dots,N-1\}}$, that maps any $x  \in  [\alpha_{i}-1/4, \alpha_{i}+1/4]$
 and $y \in [j-1/4,j+1/4]$ to a real at a distance at most $2^{-m}$ of $V_{i,j}$, for all $i \in \{1,\dots,n\}$, $j \in \{0,\dots,N-1\}$.
\end{lemma}

\newcommand\preuvesendpairs{
We start by proving the following lemma talking about some ideal sigmoids: 

\begin{lemma} \label{lem:switch:pairs}
Let $N$ be some integer. 
Assume some integers $\alpha_{1},  \alpha_{2},\dots, \alpha_{n}$, and  some constants $V_{i,j}$ for $1 \le i \le n$, and $0 \le j < N$ are given. Then there is some  function in $\manonclasslighttanh$, that we write 
$\send({(\alpha_{i},j) \sendsymbol V_{i,j}})_{i \in \{1,\dots,n\}, j \in \{0,\dots,N-1\}}$, that maps any $x  \in  [\alpha_{i}-1/4, \alpha_{i}+1/4]$
 and $y \in [j-1/4,j+1/4]$ to $V_{i,j}$, for all $i \in \{1,\dots,n\}$, $j \in \{0,\dots,N-1\}$.
\end{lemma}

\begin{proof}
If we define the function $$\send({(\alpha_{i},j) \sendsymbol V_{i,j}})_{i \in \{1,\dots,n\}, j \in \{1,\dots,N\}}(x,y)$$  by 
$\send(N\alpha_{i}+j \sendsymbol V_{i,j})_{i \in \{1,\dots,n\}, j \in \{1,\dots,N\}}(Nx+y)
$
this works when $x=\alpha_{i}$ for some $i$.
Considering instead $\send(N\alpha_{i}+j \sendsymbol V_{i,j})_{i \in \{1,\dots,n\}, j \in \{1,\dots,N\}}$ $(N
\send(\alpha_{i} \sendsymbol \alpha_{i})_{i \in \{1,\dots,n\}}(x)+y)$ works for any $x  \in  [\alpha_{i}-1/4, \alpha_{i}+1/4]$.

\end{proof}

We then go to versions with $\tanh$:

\begin{proof}[Proof of Lemma \ref{lem:switch:pairs:tanh}]
If we define the function $$\sendtanh(2^{m},{(\alpha_{i},j) \sendsymbol V_{i,j}})_{i \in \{1,\dots,n\}, j \in \{1,\dots,N\}}(x,y)$$ by $\sendtanh(2^{m},N\alpha_{i}+j \sendsymbol V_{i,j})_{i \in \{1,\dots,n\}, j \in \{1,\dots,N\}}(Nx+y)$ this works when $x=\alpha_{i}$ for some $i$.
Considering instead $$\sendtanh(2^{m},N\alpha_{i}+j \sendsymbol V_{i,j})_{i \in \{1,\dots,n\}, j \in \{1,\dots,N\}}(N
\sendtanh(2^{m+c}, \alpha_{i} \sendsymbol \alpha_{i})_{i \in \{1,\dots,n\}}(x)+y)$$ that works for any $x  \in  [\alpha_{i}-1/4, \alpha_{i}+1/4]$, 
for some constant $c$ selected as above. 
\end{proof}
}
\CORPSPREUVEDE{\preuvesendpairs}

\section{Simulating Turing machines with functions of $\manonclasslighttanh$ }
\label{sec:simulatingmt}

This section is devoted to the simulation of a Turing machine using some  analytic functions and in particular functions from $\manonclasslighttanh$. We use some ideas from \cite{BlancBournezMCU22vantardise} but with several improvements, as we need to deal with errors and avoid multiplications.

Consider without loss of generality some Turing machine $M= (Q, \{0,1,3\}, q_{\init}, \delta, F)$ using the  symbols $0,\symboleun,\symboledeux$, where $B=0$ is the blank symbol. 
\begin{remark} The reason for the choice of symbols $\symboleun$ and $\symboledeux$ will be made clear later. 
\end{remark}
 We assume $Q=\{0,1,\dots,|Q|-1\}$.  Let 
$ \dots  l_{-k} l_{-k+1} \dots l_{-1} l_{0} r_0 r_1 \dots r_n .\dots$
denote the content of the tape of the Turing machine $M$. In this representation, the head is in front of symbol $r_{0}$, and $l_i, r_{i} \in  \{0,\symboleun,\symboledeux\}$ for all $i$. 
Such a configuration $C$ can be denoted by $C=(q,l,r)$, where $l,r \in \Sigma^{\omega}$ are words over alphabet $\Sigma=\{0, \symboleun,\symboledeux\}$ and $q \in Q$ denotes the internal state of $M$. Write: $\encodagemot: \Sigma^{\omega} \to \R$ for the function that maps a word $w=w_{0} w_{1} w_{2} \dots$ to the dyadic
$\encodagemot(w)= \sum_{n \geq 0} w_n \base^{-(n+1)}$.

The idea is that such a configuration $C$ can also be encoded by some element $\bar C=(q, \bar l,\bar r) \in \N \times \R^{2}$, by considering $\bar r= \encodagemot(r)$ and $\bar l=\encodagemot(l)$. 
In other words,  we encode the configuration of a bi-infinite tape Turing machine $M$ by real numbers using their radix \base{}  encoding, but using only digits $\symboleun$,$\symboledeux$.
Notice that this lives in $Q \times [0,1]^{2}$. Denoting the image of $\encodagemot: \Sigma^{\omega} \to \R$ by $\Image$, this even lives in $Q \times \Image^{2}$. 
\begin{remark}
Notice that $\Image$ is a Cantor-like set: it corresponds to the rational numbers that can be written using only $1$ and $3$ in base $4$. We write $\Image_{S}$ for those with at most $S$ digits after the point (i.e. of the form $n/4^{S}$ for some integer $n$).
\end{remark}

\shortermcu{
A key point is to observe that }

\begin{lemma}
We can construct some function $\bar {\Next}$ in $\manonclasslighttanh$ that simulates one step of $M$: given a configuration $C$, writing $C'$ for the next configuration, we have for all integer $m$, $\| \bar{\Next}(2^{m},\bar C) - \bar C' \| \le 2^{-m}$.
\end{lemma}

\begin{proof}
We can write  $l = l_0 l^\bullet $ and $r = r_0r^\bullet $, where $l_{0}$ and $r_{0}$ are the first letters of $l$ and $r$, and $l^\bullet$ and $r^\bullet$  corresponding to the (possibly infinite) word  $l_{-1} l_{-2} \dots$ and $r_{1} r_{2} \dots$ respectively.

	\begin{center}
	\begin{tabular}{c c|c|c|c c}
		\hline 
		... & $l^\bullet $ & $l_0 $ & $r_0$ & $ r^\bullet$ & ... \\ 
		\hline 
		\multicolumn{1}{c}{} & 
		\multicolumn{2}{@{}c@{}}{$\underbrace{\hspace*{\dimexpr6\tabcolsep+3\arrayrulewidth}\hphantom{012}}_{l}$} & 
		\multicolumn{2}{@{}c@{}}{$\underbrace{\hspace*{\dimexpr6\tabcolsep+3\arrayrulewidth}\hphantom{3}}_{r}$}
	\end{tabular} 	
	\end{center}


The function $ {\Next}$ is of the form $\mathit{\Next}(q,l,r) = \mathit{\Next}(q,l^\bullet l_0,r_0r^\bullet) = (q', l', r')$ defined as a definition by case:
\[ (q', l', r') = 
\left\{ \begin{array}{ll}
(q', l^\bullet l_0 x, r^\bullet) & \mbox{whenever $\delta(q,r_{0}) = (q',x, \rightarrow)$} \\
(q', l^\bullet, l_0 x r^\bullet) & \mbox{whenever $\delta(q,r_{0}) = (q',x, \leftarrow)$} \\
\end{array} \right.
\]

This provides a first candidate for the function $\bar{\Next}$: Consider the similar function working over the representation of the configurations as reals, considering $r_{0} =  \lfloor \base \bar{r}\rfloor$ 
\begin{align*}
\mathit{\bar{\Next}}(q,\bar l, \bar r) &= \mathit{\bar{\Next}}(q,\bar{l^\bullet l_0},\bar{r_0r^\bullet}) = (q', \bar{l'}, \bar{r'})\\
&= \left\{ 
\begin{array}{ll}
	(q', \bar{l^\bullet l_0 x}, \bar{r^\bullet}) & \mbox{whenever $\delta(q,r_{0}) = (q',x, \rightarrow)$} \\
  (q', \bar{l^\bullet}, \bar{l_0 x r^\bullet}) & \mbox{whenever $\delta(q,r_{0}) = (q',x, \leftarrow)$} \\
\end{array}
\right.
\end{align*} 

\begin{equation}
\begin{array}{l} \label{textaremplacer}
\mbox{
	$\bullet$  in the first case ``$\rightarrow$'' : $\bar{l'} = \base^{-1} \bar l + \base^{-1} x $ and $\bar{r'}  = \bar{r^\bullet} = \{\base \bar r\} $} \\
\mbox{
	$\bullet$ in the second case ``$\leftarrow$'' : $\bar{l'} =\bar{ l^\bullet} = \{\base \bar l \} $ and $\bar{r'} = \base^{-2} \{\base \bar r\}  + \base^{-2} x + \lfloor \base \bar{l}\rfloor / 4$ }\end{array}
\end{equation}

We introduce the following functions: $\rightarrow : Q \times \{0,1,3\} \mapsto \{0,1\}$ and $\leftarrow:Q \times \{0,1,3\} \mapsto \{0,1\}$ such that $\rightarrow(q,a)$ (respectively: $\leftarrow(q,a)$)
is $1$ when $\delta(q,a) = (\cdot,\cdot, \rightarrow)$ (resp. $(\cdot,\cdot, \leftarrow)$), i.e. the head moves right (resp. left), and $0$ otherwise. We define  $nextq^{q}_{a}=q'$ if $\delta(q,a)= (q',\cdot,\cdot)$, i.e. values $(q',x,m)$ for some $x$ and  $m \in \{\leftarrow,\rightarrow\}$.

\manonplusimportant{commenté: 
\[ \rightarrow(q,a) = \left\{ 
\begin{array}{ll}
1 & \mbox{if $\delta(q,a) = (\cdot,\cdot, \rightarrow) \in \Delta$} \\
0 & \mbox{else}
\end{array}
\right.
\]
and 
\[ \leftarrow(q,a) = \left\{ 
\begin{array}{ll}
1 & \mbox{if $\delta(q,a) = (\cdot,\cdot, \leftarrow) \in \Delta$} \\
0 & \mbox{else}
\end{array}
\right.
\]
For $D = \{\rightarrow, \leftarrow \} $, }
We can rewrite $ \mathit{\bar{\Next}}(q,\bar l, \bar r) = (q',\bar l',\bar r')$ as
$$\bar{l'} = \displaystyle \sum_{q,r_{0}} \left[\rightarrow(q, r_{0}) \left( \frac{\bar{l}}{4} + \frac{x}{4} \right) + \leftarrow(q, r_{0}) \left\{ 4 \bar{l}\right\}\right] $$
and
$$\bar{r'} = \displaystyle \sum_{q,r_{0}} \left[\rightarrow(q, r_{0}) \left\{ 4 \bar{r} \right\} + \leftarrow(q, r_{0}) \left( \frac{\left\{ 4r  \right\} }{4^2} + \frac{x}{4^2} + \frac{\lfloor 4 \bar{l} \rfloor }{4}\right)\right],$$
and, using notation of Lemma \ref{lem:switch:pairs:tanh},
$q'=\send((q,r) \sendsymbol nextq^{q}_{r})_{q \in Q, r\in \{0,1,3\}}(q,\lfloor 4 \bar{r} \rfloor ).$

\manonplusimportant{commenté: 
$$\bar{l'} = \sum_{\left. \begin{array}{ll}
	~~~~~~~~q\in Q \\
	\delta(q,r_0) = (q',x,D)\\
	\end{array}
	\right. } \rightarrow(q, r_0) \left( \frac{\bar{l}}{4} + \frac{x}{4} \right) + \leftarrow(q, r_0) \left\{ 4 \bar{l}\right\} $$

$$\bar{r'} =  \sum_{ \left. \begin{array}{ll}
	~~~~~~~~q\in Q \\
	\delta(q,r_0) = (q',x,D)\\
	\end{array}
	\right. } \rightarrow(q, r_0) \left\{ 4 \bar{r} \right\} + \leftarrow(q, r_0) \left( \frac{\left\{ 4r \right\} }{4^2} + \frac{x}{4} + \lfloor 4 \bar{l} \rfloor \right) $$
}

Our problem with such expressions is that they involve some discontinuous functions such as the integer part and the fractional part function, and we would rather have analytic (hence continuous) functions.
However, a key point is that from our trick of using only symbols $\symboleun$ and $\symboledeux$, we are sure
that in an expression like $\lfloor 4 \bar r \rfloor$, either it values $0$ (this is the specific case where there remain only blanks in $r$), or that $\base \bar r$ lives in an interval
$[\symboleun,2]$ or in interval $[\symboledeux,4]$. 

That means that we could replace $\lfloor \base \bar r \rfloor$ by $\sigma(\base \bar r)$
 if we take $\sigma$ as some continuous function that would be affine and values respectively $0$, $1$ and $3$ on $\{0\} \cup [\symboleun,2] \cup [\symboledeux,4]$
 (that is to say matches $\lfloor \base \bar r \rfloor$ on this domain). A possible candidate is $\sigma(x)=\sig(1/4,3/4,x) + \sig(9/4,11/4,x)$.  
  Then considering $\xi(x)=x-\sigma(x)$, then  $\xi(\base \bar r)$ would be the same as $\left\{ \base \bar r \right\}$:  
that is, considering  $r_{0}= \sigma(4\bar{r})$,  replacing in the above expression every $\left\{ \base \cdot \right\}$ by $\xi(\cdot)$, and every $\lfloor \cdot \rfloor$ by $\sigma( \cdot)$, and get something that would still work the same, but using only 
continuous functions.

But,  we would like to go to some analytic functions and not only continuous functions, and it is well-known that an analytic function that equals some affine function on some interval (e.g. on [1,2]) must be affine, and hence cannot be $3$ on $[3,4]$. But the point is that we can try to tolerate errors, and replace $\sig(\cdot,\cdot)$ by $\sigtanh(2^{m+c},\cdot,\cdot)$ in the expressions above for $\sigma$ and $\xi$, taking $c$ such that $(3+1/4^{2})3|Q|  \le 2^{c}$. This would just introduce some error at most  $(3+1/4^{2})3|Q| 2^{-c}2^{-m}\le 2^{-m}$.

\olivierplusimportant{Mieux expliquer cela: il faudrait etre plus explicite}
\begin{remark}  We could also replace every $\rightarrow(q,r)$ in above expressions for $\bar l'$ and $\bar r'$ by $\sendtanh(k, (q,r) \sendsymbol \rightarrow(q,r))(q,\sigma(4 \bar r))$, for a suitable error bound $k$, and symmetrically for $\leftarrow(q,r)$.
However, if we do so, we still might have some multiplications in the above expressions. 
\end{remark}
The key is to use Lemma \ref{tricksigmoidtanh}: we can also write the above expressions as
$$
\begin{array}{ll}
\bar{l'} =   \sum_{q,r}& \left[ 
\TTtanh\left(2^{m+c},\sendtanh(2^{2},(q,r) \sendsymbol \rightarrow(q,r)) (q,\sigma(4 \bar r)), 
 \frac{\bar{l}}{4} + \frac{x}{4} \right) \right. \\
&+ \left.
\TTtanh\left(2^{m+c},\sendtanh(2^{2},(q,r) \sendsymbol \leftarrow(q,r)) (q,\sigma(4 \bar r)),
\xi( 4 \bar{l}) \right) \right]
\end{array}
$$
$$
\begin{array}{ll}
\bar{r'} =   \sum_{q,r}& \left[ 
\TTtanh\left(2^{m+c},\sendtanh(2^{2},(q,r) \sendsymbol \rightarrow(q,r)) (q,\sigma(4 \bar r)),
\xi( 4 \bar{r} ) \right) \right. \\
&+ \left.
\TTtanh\left(2^{m+c},\sendtanh(2^{2},(q,r) \sendsymbol \leftarrow(q,r)) (q,\sigma(4 \bar r)),
\frac{\xi( 4r )}{4^2} + \frac{x}{4^2} + \frac{\sigma(4 \bar{l})}{4} \right) \right]
\end{array}
$$
and still have the same bound on the error.
\end{proof}

Once we have one step, we would like to simulate some arbitrary computation of a Turing machine, by considering the iterations of function $\Next$. 

The problem with the above construction is that even if we start from the exact encoding $\bar C$ of a configuration, it  introduces some error (even if at most $2^{-m}$). If we want to apply again the function $\Next$, then we will start not exactly from the encoding of a configuration. Looking at the choice of the function $\sigma$, a small error can be tolerated (roughly if the process does not involve points at a distance less than $1/4$ of $\Image$), but this error is amplified (roughly multiplied by $4$ on some component), before introducing some new errors (even if at most $2^{-m}$).  The point is that if we repeat the process, very soon it will be amplified, up to a level where we have no true idea or control about what becomes the value of the above function. 

However, if we know some bound on the space used by the Turing machine, we can correct it to get at most some fixed additive error: a Turing machine using a space $S$ uses at most $S$ cells to the right and to the left of the initial position of its head. Consequently, a configuration $C=(q,l,r)$ of such a machine involves words $l$ and $r$ of length at most $S$. Their encoding $\bar l$, and $\bar r$ are expected to remain in $\Image_{S+1}$. Consider $\round_{S+1}(\overline{l})=\lfloor 4^{S+1} \bar l   \rfloor / 4^{S+1}.$ For a point $\overline{l}$ of $\Image_{S+1}$,  $4^{S+1} \bar l$ is an integer, and $\overline{l} = \round_{S+1}(\overline{l})$. But now, for a point $\tilde{\overline{l}}$ at a distance less than $4^{-(S+2)}$ from a point $\overline{l} \in \Image_{S+1}$,  $\round_{S+1}(\tilde{\overline{l}})=\overline{l}$. In other words,
$\round_{S+1}$ ``deletes'' errors of order $4^{-(S+2)}$.
Consequently, we can replace every $\bar l$ in the above expressions by $\sigma_1(2^{2S+4},2^{2 S+3}, 4^{S+1} \bar l) / 4^{S+1},$
as this is close to $\round_{S+1}(\bar l)$, and the same for $\bar r$, 
where $\sigma_{1}$ is the function from Corollary \ref{lem:i}. 
We could also replace $m$ by $m+2S+4$ to guarantee that $2^{-m} \le 4^{-(S+2)}$. 
We get the following important improvement of the previous lemma:

%

\begin{lemma}
We can construct some function $\bar {\Next}$ in $\manonclasslighttanh$ that simulates one step of $M$, i.e. that computes the $\Next$ function sending a configuration $\bar C$ of Turing machine $M$ to $\bar C'$, where $C'$ is the next one:  $\| \Next(2^{m}, 2^{S},\bar C) - \bar C' \| \le 2^{-m}$.  

Furthermore, it is robust to errors on its input, up to space $S$: considering $\|\tilde{C}-\bar C\| \le 4^{-(S+2)}$, $\| \Next(2^{m}, 2^{S},\tilde C) - \bar C' \| \le 2^{-m}$ remains true.
\end{lemma}

\begin{proposition} \label{prop:deux} 
Consider some Turing machine $M$ that computes some function $f: \Sigma^{*} \to \Sigma^{*}$ in some time $T(\ell(\omega))$ on input $\omega$.  One can construct some function $\tilde{ f}: \N^{2} \times \R \to \R$ in $\manonclasslighttanh$ that does the same: $\tilde{ f}(2^{m}, 2^{T(\ell(\omega))},
\encodagemot(\omega))$ that is at most $2^{-m}$ far from $\encodagemot(f(\omega))$.
\end{proposition}

\begin{proof}
The idea is to define the function $\bar {Exec}$ that maps some time $2^{t}$ and some initial configuration $C$ to the configuration  at time $t$.  This can be obtained  using previous lemmas by 
$$
\left\{
\begin{array}{lll}
 \bar {Exec}(2^{m},0,2^{T},C)  &=&  C  \\
\bar {Exec} 
  (2^{m}, 2^{t+1},  2^{T},C) 
 &=& \bar{\Next}(2^{m},2^{T},\bar {Exec}(2^{m},2^{t},2^{T},C)).
 \end{array}\right
 .$$
	
We can then get the value of the computation as $\bar {Exec}(2^{m}, 2^{T(\ell(\omega))},2^{T(\ell(\omega))},C_{\init})$
	on input $\omega$, considering $C_{\init}=(q_{0},0,\encodagemot(\omega))$.
	By applying some projection, we get the following function
	$\tilde{ f}(2^{m},2^{T},y)= \projection{3}{3}(\bar {Exec}(2^{m}, 2^{T},2^{T}, (q_{0},0,y)))$ that satisfies the property.	\end{proof}

	
Actually, in order to get $\FPspace$, observe that we can also replace the linear length ODE by a linear ODE.

\begin{proposition} \label{prop:trendeux} 
Consider some Turing machine $M$ that computes some function $f: \Sigma^{*} \to \Sigma^{*}$ in some polynomial space $S(\ell(\omega))$ on input $\omega$.  One can construct some function $\tilde{ f}: \N^{2} \times \R \to \R$ in $\spaceclasstanh$ that does the same: we have $\tilde{ f}(2^{m}, 2^{S(\ell(\omega))},
\encodagemot(\omega))$ that is at most $2^{-m}$ far from $\encodagemot(f(\omega))$.
\end{proposition}

\begin{proof}
The idea is the same, but not working with powers of $2$, and with linear ODE:  define the function $\bar {Exec}$ that maps some time $t$ and some initial configuration $C$ to the configuration  at time $t$.  This can be obtained the using previous lemma by 
$$
\left\{
\begin{array}{lll} 
 \bar {Exec}(2^{m},0,2^{S},C) &=& 
 C \\
\bar {Exec} 
  (2^{m}, {t+1},  2^{S},C) 
 &=& \bar{\Next}(2^{m},2^{S},\bar {Exec}(2^{m},{t},2^{S},C)).
 \end{array}\right.
 $$
 
 In order to claim this is a robust linear ODE, we need to state that  ${Exec} 
  (2^{m}, {t},  2^{S},C)$ is polynomially numerically stable: but this holds, since to estimate this value at $2^{-n}$  it is sufficient to work at precision $4^{-max(m,n,S+2)}$ (independently of $t$, from the  rounding).
   
We can then get the value of the computation as $\bar {Exec}(2^{m}, 2^{S(\ell(\omega))},2^{S(\ell(\omega))},C_{\init})$
	on input $\omega$, considering $C_{\init}=(q_{0},0,\encodagemot(\omega))$.
	By applying some projection, we get the following function
	$\tilde{ f}(2^{m},2^{S},y)= \projection{3}{3}(\bar {Exec}(2^{m}, S,2^{S}, (q_{0},0,y)))$ that satisfies the property.	
\end{proof}

\section{Converting integers and dyadics to words, and conversely}
\label{sec:convertion}

One point of our simulations of Turing machines is that they work over $\Image$, through encoding $\encodagemot$, while we would like to talk about integers and real numbers: we need to be able to convert an integer (more generally a dyadic) into some  encoding over $\Image$ and conversely. 

Fix the following encoding: every digit in the binary expansion of $d$  is encoded by a pair of symbols in the radix $4$ expansion of $\overline{d} \in \Image \cap [0,1]$: digit $0$ (respectively: $1$) is encoded by $11$ (resp. $13$) if before the ``decimal'' point in $d$, and digit $0$ (respectively: $1$) is encoded by $31$ (resp. $33$) if after. For example, for $d=101.1$ in base $2$, $\overline{d}=0.13111333$ in base $4$. 


\begin{lemma}[{From $\N$ to $\Image$}]  \label{lem:manquant} We can construct some function $\Decode: \N^{2} \to \R$ in $\manonclasslighttanh$ that maps $m$ and $n$ to some point at distance less than $2^{-m}$ from $\encodagemot(\overline{n})$.
\end{lemma}

\newcommand\preuvetrenteneuf{
\begin{proof} \label{trenteneuf} Recall the functions provided by Corollaries \ref{lem:mod2} and \ref{lem:div2}.
The idea is to iterate $\ell(n)$ times function 
$$
F( \overline{r_1}, \overline{l_2}) = 
\left\{
\begin{array}{ll} 
(\div_{2}(\overline{r_1}), (\overline{l_2} + 5)/4) &  \mbox{ whenever } \mod_{2}( \overline{r_1})=0 \\
(\div_{2}(\overline{r_1}), (\overline{l_2} + 7)/4) &  \mbox{ whenever } \mod_{2}( \overline{r_1})=1. \\
\end{array}\right. 
$$
over $(n,0)$, and then projects on the second argument.   

This can be done in $\manonclasslighttanh$ by considering 
$F(2^{m},2^{M},\overline{r_1}, \overline{l_2}) =\sendtanh(2^{m+1},$ \\
$0 \sendsymbol (\div_{2}(2^{m+1},2^{M},\overline{r_1}), (\overline{l_2} + 5)/4),$ \\
$1 \sendsymbol (\div_{2}(2^{m+1},2^{M},\overline{r_1}), (\overline{l_2} + 7)/4)) (( \overline{r_1}),$ and then
$$\Decode(2^{m},n)= \projection{2}{2}(G(2^{m+\ell(n)},2^{\ell(n)+1},2^{\ell(n)},n,0)),$$ with

$$\left\{ 
\begin{array}{lll} 
G(2^{m},2^{M},2^{l},2^{0},n,0) &=& (n,0) \\
 G (2^{m}, 2^{M},2^{l+1},r,l)  
 &=&F(2^{m},2^{M},G(2^{m},2^{l},r,l)).
 \end{array}\right.$$
 The global error will be at most $2^{-m-\ell(n)} \times \ell(n) \le 2^{-m}$.
\end{proof}
}
\CORPSPREUVEDE{\preuvetrenteneuf}

This technique can be extended to consider decoding of tuples: there is a function $\Decode: \N^{d+1} \to \R$ in $\manonclasslighttanh$ that maps $m$ and $\tu n$ to some point at distance less than $2^{-m}$ from $\encodagemot(\overline{\tu n})$, with $\overline{\tu n}$ defined componentwise.

Conversely, given $\overline{d}$, we need a way to construct $d$. Actually, as we will need to avoid multiplications, we state that we can even do something stronger: given $\overline{d}$, and (some bounded) $\lambda$ we can construct $\lambda d$.

\begin{lemma}[{From $\Image$ to $\R$}, and multiplying in parallel]  \label{lem:codage:manon}
We can construct some  function $\EncodeMul: \N^{2} \times [0,1] \times \R \to \R$ in $\manonclasslighttanh$ that  maps $m$, $2^{S}$, $\encodagemot(\overline{d})$  and (bounded) $\lambda$ to some real at distance at most $2^{-m}$ from  $\lambda d$, whenever  $\overline{d}$ is of length less than $S$.
\end{lemma}

\newcommand\preuvelemcodagemanon{
\begin{proof}
The idea is to do as in the proof of previous lemma, but considering 
$$
F( \overline{r_1}, \overline{l_2},\lambda) = 
\left\{
\begin{array}{ll} (\sigma(16 \overline{r_1}), 2 \overline{l_2} + 0,\lambda) &  \mbox{ whenever } i( 16 \overline{r_1})= 
5\\
(\sigma(16 \overline{r_1}), 2 \overline{l_2} + \lambda,\lambda) &  \mbox{ whenever } i( 16 \overline{r_1})= 
7\\
 (\sigma(16 \overline{r_1}), (\overline{l_2} + 0)/2,\lambda) &  \mbox{ whenever } i( 16 \overline{r_1})= 
 13\\
 (\sigma(16 \overline{r_1}), (\overline{l_2} + \lambda)/2,\lambda) &  \mbox{ whenever } i( 16 \overline{r_1})= 
 15
\end{array}\right.
$$
iterated $S$ times over suitable approximation of the rounding $\round_{S+1}(\encodagemot(\overline{d}),0,\lambda))$, with  $\sigma$ and $\xi$ 
constructed as  approximation of the integer and fractional part, as before. 
\end{proof}
}
\CORPSPREUVEDE{\preuvelemcodagemanon}
		
\section{Proofs and applications}
\label{sec:applications}

  Theorem \ref{trucchosethmain}  follows from  point 1. of next Proposition for one inclusion, and previous simulation of Turing machines for the other.

\MFCSshorter{
We start by stating the following  (similar to the statement about $\manonclass$ in \cite{BlancBournezMCU22vantardise}):}



\begin{proposition}  \label{prop:mcu:un}
	\begin{enumerate}
\item  All functions of $\manonclasslighttanh$  are computable (in the sense of computable analysis) in polynomial time.

\item  All functions of $\spaceclass$  are computable (in the sense of computable analysis) in polynomial space.
\end{enumerate}
\end{proposition}

%
%
%
%
%


\begin{proof}
The proof  consists in observing this holds for the basic functions and that composition preserves polynomial time (respectively: space) computability and also by linear length ODEs. This latter fact is established by computable analysis arguments, reasoning on some explicit formula giving the solution of linear length ODE. The proof is similar to the statement about $\manonclass$ in \cite{BlancBournezMCU22vantardise}. In order to be self-content, we repeat in Appendix \ref{ouestlapreuve}, how this is established in \cite{BlancBournezMCU22vantardise}.

Regarding space, the main issue is the need to prove the schema given by  Definition \ref{schema:space} guarantees $\tu f$ is in $\FPspace$ when $\tu u$, $\tu g$, and $\tu h$ are. Assuming condition 1. of Definition \ref{schema:space} would not be sufficient: the problem is that $\tu f(x,\tu y)$ may polynomially grow too fast or have a modulus function that would grow too fast. The point is, in Definition \ref{schema:space}, we assumed $\tu f$ to be both bounded and satisfying \textbf{2.}, i.e. polynomial numerical robustness. With these hypotheses, it is sufficient to work with the precision given by this robustness condition and these conditions guarantee the validity of computing with such approximated values.
\end{proof}

We now go to various applications of the proposition and of our toolbox. First, we state a characterisation of $\FPtime$ for general functions, covering both the case of a function $\tu f: \N^{d} \to \R^{d'}$, $\tu f: \R^{d} \to \R^{d'}$ as a special case: only the first type (sequences) was covered by \cite{BlancBournezMCU22vantardise}.

\begin{theorem}[Theorem \ref{th:main:twop}] \label{th:main:one:ex:p}
A function  $\tu f: \R^{d}  \times \N^{d''} \to \R^{d'}$ is computable in polynomial time iff there exists $\tilde{\tu f}:\R^{d} \times \N^{d''+2} \to \R^{d'} \in \manonclasslighttanh$ such that
for all $\tu x \in \R^{d}$, $X \in \N$, $ \tu x \in\left[-2^{X}, 2^{X}\right]$, $\tu m  \in \N^{d''}$, $n \in \N$,
$\|\tilde{\tu f}(\tu x, \tu m,2^{X},2^{n}) - \tu f(\tu x, \tu m) \| \le 2^{-n}.$
\end{theorem}

The reverse implication of Theorem \ref{th:main:one:ex:p} follows from Proposition \ref{prop:mcu:un}, (1.) and arguments from computable analysis.
%
%
\newcommand\preuvereversethmainoneexp{
\begin{proof}[Proof of reverse implication of Theorem \ref{th:main:one:ex:p}]
	Assume there exists $\tilde{\tu f}:\R^{d} \times \N^{d''+2} \to \R^{d'} \in \manonclasslighttanh$ such that
	for all $\tu x \in \R^{d}$, $X \in \N$, $ \tu x \in\left[-2^{X}, 2^{X}\right]$, $\tu m  \in \N^{d''}$, $n \in \N$,
	$\|\tilde{\tu f}(\tu x, \tu m,2^{X},2^{n}) - \tu f(\tu x, \tu m) \| \le 2^{-n}.$
	
	From Proposition  \ref{prop:mcu:un}, (1.), we know that $\tilde{\tu f}$ is computable in polynomial time (in the binary length of its arguments).
	Then $\tu f(\tu x,\tu m)$ is computable: indeed, given $\tu x$, $\tu m$ and $n$, we can approximate $\tu f(\tu x, \tu m)$ at precision $2^{-n}$ on $[-2^{X}, 2^{X}]$ as follows: approximate $\tilde{\tu f}(\tu x, \tu m,2^{X},2^{n+1})$ at precision $2^{-(n+1)}$ by some rational $q$, and output $q$. We will then have
	$$
	\begin{array}{ll}
		\|q-\tu f(\tu x, \tu m)\| &\le \|q-\tilde{\tu f}(\tu x, \tu m,2^{X},2^{n+1}) \| + \|\tilde{\tu f}(\tu x, \tu m,2^{X},2^{n+1})-\tu f(\tu x, \tu m)\| \\ &\le 2^{-(n+1)} + 2^{-(n+1)} \\ &\le 2^{-n}.
	\end{array}
	$$
	All of this is done in polynomial time in $n$ and the size of $\tu m$, and hence we get that $\tu f$ is polynomial time computable from definitions.
\end{proof}
}
\CORPSPREUVEDE{\preuvereversethmainoneexp}
For the direct implication, for sequences, that is to say, functions of type $\tu f: \N^{d''} \to \R^{d'}$ (i.e. $d=0$, the case considered in \cite{BlancBournezMCU22vantardise}) we are almost done: reasoning componentwise, we only need to consider  $ f: \N^{d''} \to \R$ (i.e. $d'=1$).  As the function is polynomial time computable, this means that there is a  polynomial time computable function $g: \N^{d''+1} \to \{\symboleun,\symboledeux\}^{*}$ so that on $\tu m,2^{n}$, it provides the encoding $\bar{\phi(\tu m,n)}$ of some dyadic  $\phi(\tu m,n)$ with $\|\phi(\tu m,n)-\tu f(\tu m)\| \le 2^{-n}$ for all $\tu m$.  
The problem is then to decode, compute and encode the result to produce this dyadic, using our previous toolbox. 

More precisely, from Proposition \ref{prop:deux},  we get $\tilde{g}$ with  
 $$|\tilde{g}(2^{e},2^{p(max(\tu m,n))}, \Decode(2^{e},\tu m,n)) -\encodagemot(g(\tu m,n)) | \le 2^{-e}$$ for some polynomial $p$ 
corresponding to the time required to compute  $g$, and $e=\max(p(max(\tu m,n)),n)$. Then we  need to transform the  value to the correct dyadic: we mean 
   $$\tilde{\tu f}(\tu m,n)=\EncodeMul(2^{e},2^{t},\tilde{g}(2^{e},2^{t},\Decode(2^{e},\tu m,n)),1),$$ where $t=p(max(\tu m,n))$,   $e=\max(p(max(\tu m,n)),n)$ provides a solution such that
 $\|\tilde{\tu f}(\tu m,2^{n})-\tu f(\tu m)\| \le 2^{-n}.$

\begin{remark} This is basically what is done  in  \cite{BlancBournezMCU22vantardise}, except that we do it here with  analytic functions. 
However, as already observed in  \cite{BlancBournezMCU22vantardise}, this cannot be done for the case $d \ge 1$, i.e. for example for $f: \R \to \R$. The problem is that we used the fact that we can decode: $\Decode$ maps an integer $n$ to its encoding $\overline{n}$ (but is not guaranteed to do something valid on non-integers).  There cannot exist such functions that would be valid over all reals, as such functions must be continuous, and there is no way to map continuously real numbers to finite words. This is where the approach of the article \cite{BlancBournezMCU22vantardise} is stuck. 
\end{remark}

To solve this, we use an adaptive barycentric technique.  For simplicity, and pedagogy, we discuss only the case of a polynomial time computable function $f: \R \times \N \to \R$. 
From standard arguments from computable analysis (see e.g. [Corollary 2.21]\cite{Ko91}), the following holds and the point is to be able to realise all this with functions from $\manonclasslighttanh$.
\begin{lemma}
	Assume $f: \R \times \N \to \R$ is computable in polynomial time. There exists some polynomial $m: \N^{2} \to \N$ and  some $\tilde{f}: \N^{4} \to \Z$ computable in polynomial time
	such that for all $x\in \R$, $\| 2^{-n} \tilde{f}(\lfloor 2^{m(n,M)}  x  \rfloor, u, 2^{M}, 2^{n}) - f(x,u) \| \leq 2^{-n}$ whenever $\frac{x}{2^{m(n,M)}} \in [-2^{M},2^{M}]$.
\end{lemma}

Assume we consider an approximation $\sigma_{i}$ (with either $i=1$ or $i=2$) of the integer part function given by Lemma \ref{lem:i}. Then, given $n,M$, when $2^{m(n,M)} x$ falls in some suitable interval $I_{i}$ for $\sigma_{i}$ (see the statement of Lemma \ref{lem:i}), we are sure that $\sigma_{i}(2^{e},2^{m(n,M)+X+1}, 2^{m(n,M)}  x)$ is at some  distance upon control from $\lfloor 2^{m(n,M)}  x  \rfloor$. Consequently, $2^{-n} \tilde{f}(\sigma_{i}(2^{m(n,M)+X+1}, 2^{m(n,M)} $ $ x ), u, 2^{M}, 2^{n})$ provides some $2^{-n}$-approximation of $f(x,u)$, up to some  error upon control. 
When this holds, we then use an argument similar to what we describe for sequences:  using functions from $\manonclasslighttanh$, we can decode, compute, and encode the result to provide this dyadic. It is provided by an expression
$\Formula_{i}(x,u,M,n)$ of the form 
$\EncodeMul(2^{e},2^{t},\tilde{\tilde{f}}(2^{2},2^{t},$ $\Decode(2^{e},
\sigma_{i}(2^{e},2^{M}, 2^{m(n,M)}  x )
)),2^{-n})$.

The problem is that it might also be the case that $2^{m(n,M)}  x$ falls in the complement of the intervals $(I_{i})_{i}$. In that case, we have no clear idea of what could be the value of $\sigma_{i}(2^{e},2^{m(n,M)+X+1}, 2^{m(n,M)}  x )$, and consequently of what might be the value of the above expression $Formula_{i}(x,u,$ $M,n)$. 
But the point is that when it happens for an $x$ for $\sigma_{1}$, we could have used $\sigma_{2}$, and this would work, as one can check that the intervals of type $I_{1}$ cover the complements of the intervals of type $I_{2}$ and conversely. They also overlap, but when $x$ is both in some $I_{1}$ and $I_{2}$, $Formula_{1}(x,u,M,n)$ and $\Formula_{2}(x,u,M,n)$ may differ, but they are both $2^{-n}$ approximations of $f(x)$.

The key is  to compute some suitable "adaptive" barycenter, using function $\lambda$, provided by Corollary \ref{lem:lambda}. Writing $\approx$ for the fact that two values are closed up to some controlled bounded error, observe from the statements of Corollary \ref{lem:lambda} and \ref{lem:i}\shortermcu{	\begin{itemize}
	\item}  that whenever $\lambda(\cdot,2^n,x) \approx 0$, we know that $\sigma_{2}(\cdot,2^n,x) \approx \lfloor x \rfloor$; 
	\shortermcu{ \item that} whenever $\lambda(\cdot,2^n,x) \approx 1$ we know that $\sigma_{1}(\cdot,2^n,x) \approx \lfloor x \rfloor$;
	\shortermcu{ \item that}    whenever $\lambda(\cdot, 2^n,x) \in (0,1)$, we know that $\sigma_1(\cdot, 2^n,x) \approx \lfloor x \rfloor +1 $ and $\sigma_2(\cdot, 2^n,x) \approx \lfloor x \rfloor$.
	\shortermcu{\end{itemize}}
That means that if we consider 
$$\lambda(\cdot, 2^n,x) Formula_{1}(x,u,M,n) + (1-\lambda(\cdot, 2^n,n)) Formula_{2}(x,u,M,n) $$ we are
sure to be close (up to some bounded error) to some $2^{-n}$ approximation of $f(x)$.
There remains that this requires some multiplication with $\lambda$. But from the form of $\Formula_{i}(x,u,M,n)$, this could be also be written as follows, 
ending the proof of  Theorem \ref{th:main:one:ex:p}.

\vspace{-0.6cm}
{
\begin{multline}\label{letrucalafin}
\hspace{-0.5cm} \EncodeMul(2^{e},2^{t},\tilde{\tilde{f}}(2^{e},2^{t},\Decode(2^{e},
\sigma_{1}(2^{e},2^{M}, 2^{m(n,M)}  x )
)),\lambda(2^{e},2^{M}, 2^{m(n,M)} x )  2^{-n}) + \\
\hspace{-0.5cm} \EncodeMul(2^{e},2^{t},\tilde{\tilde{f}}(2^{e},2^{t},\Decode(2^{e},
\sigma_{2}(2^{e},2^{M}, 2^{m(n,M)}  x )
)),(1-\lambda(2^{e},2^{M}, 2^{m(n,M)} x ))
  2^{-n})
\end{multline}
}

\vspace{-0.3cm}
\newcommand\preuvetrucchoseth{
\begin{proof}[Proof of Theorem \ref{trucchosethmain}] 
	We know that a function $\tu f: \R^{d} \to \R^{d'}$ from $\manonclasslighttanh$ is polynomial time computable by Proposition \ref{prop:mcu:un}, (1.). That means that we can approximate it with arbitrary precision, in particular precision $\frac14$ in polynomial time. Given such an approximation $\tu q$, if we know it is some integer, it is easy to determine which integer it is:  return (componentwise) the closest integer to $\tu q$.
	
	Conversely, if we have a function $\tu f: \N^{d} \to \N^{d'}$ that is polynomial time computable, our previous simulations of Turing machines provide a function in  $\manonclasslighttanh$  that computes it at any required precision, in particular $1/4$.
\end{proof}
}
\CORPSPREUVEDE{\preuvetrucchoseth}

\olivierplusimportant{Commenté:
\begin{proof}
	The representation of integers being finite, we do not need, for all $n \in \N$, some arbitrary precision on the input, but only precision $2^{-m}$. Then $\tilde{\tilde{f}}$ is basically the same as function $\tilde{f}$, but replacing the oracle calls by readings on its first argument.
\end{proof}
}

From the fact that we have the reverse direction in Theorem \ref{th:main:one:ex:p}, it is natural to consider the operation that maps $\tilde{\tu f}$ to $\tu f$. Namely, we introduce the operation $\MANONlim$ (this stands for Effective Limit):

\begin{definition}[Operation $\MANONlim$] Given $\tilde{\tu f}:\R^{d} \times \N^{d''} \times \N \to \R^{d'} \in \manonclasslighttanh$ such that
for all $\tu x \in \R^{d}$,  $X \in \N$, $ \tu x \in\left[-2^{X}, 2^{X}\right]$, $\tu m \in \N^{d''}$, $n \in \N$,
$\|\tilde{\tu f}(\tu x, \tu m,2^X,2^{n}) - \tu f(\tu x, \tu m) \| \le 2^{-n},$ then 
$\MANONlim(\tilde{\tu f})$ is the (clearly uniquely defined) corresponding function  $\tu f: \R^{d} \to \R^{d'}$.
\end{definition}

\begin{theorem}  \label{th:main:two} 
A continuous function $\tu f$ 
 is  computable in polynomial time if and only if all its components belong to $\manonclasslighttanhlim$, where
$$\manonclasslighttanhlim= [\mathbf{0},\mathbf{1},\projection{i}{k},   \length{x}, \plus, \minus,  \signb{x}, \frac{x}{2}, \frac{x}{3};
 composition,  linear~length~ODE,\MANONlim].$$
\end{theorem}

For the reverse direction, by induction, the only thing to prove is that the class of functions from to integers computable in polynomial time is preserved by the operation $\MANONlim$. Take such a function $\tilde{\tu f}$.  By definition, given $\tu x, \tu m$, $X$ we can compute $\tilde{f}(\tu x, \tu m, 2^X, 2^n)$ with precision $2^{-n} $ in time polynomial in $n$. This must be, by definition of $\MANONlim$ schema, some approximation of $\tu f(\tu x, \tu m)$ over $[-2^{X},2^{X}]$, and hence $\tu f$ is computable in polynomial time. 
This also gives directly Theorem \ref{th:main:twop} as a corollary. 

\MFCSshorter{
From the proofs, we also get a normal form theorem, namely formula \eqref{letrucalafin}. In particular, \begin{theorem}
Any function $f: \N^{d} \times \R^{d''} \to \R^{d'}$ can be obtained from the class $\manonclasslighttanhlim$ using only one schema $\MANONlim$. 
\end{theorem}
}

We obtain the  statements for polynomial space computability (Theorems \ref{trucchosethmainR} and \ref{th:main:twopSpace}) replacing $\manonclasslighttanh$ by $\spaceclasstanh$, using similar reasoning about space instead of time, considering point 2. instead of 1. of Proposition  \ref{prop:mcu:un}, and Proposition \ref{prop:trendeux} instead of Proposition \ref{prop:deux}.
\MFCSshorter{
We now comment on the relations with formal neural network models. In this article, we are programming with $\tanh$ and sigmoids. We expressed the sigmoids in terms of the $\relu$ function, through Lemma \ref{lemmeRelu}. Function $\tanh$ could be replaced by $\arctan$: the key was to be able to approximate the $\relu$ function with $\tanh$ (Lemma \ref{lemmeRelu}) and this can be proved to hold also for $\arctan$, using error bounds on $\arctan$ established in \cite{TheseDaniel}.

Now, given some function $\tu f: \R^{d} \to \R^{d'}$, some error and some time $t$, our expressions provide explicit expressions in $\manonclasslighttanh$ of an approximation of what is computed by a Turing machine at time $t$ uniformly over any compact domain. 

\begin{remark}
The formula \ref{letrucalafin} can be seen as a function that generates uniformly a family of circuits/formal $\sigtanh$ approximating a given function at some given precision over some given domain. The functions we generate are the composition of essentially linear functions, which can be considered as layers of formal neural networks\footnote{With a concept of neural network that is not assuming that the last layer of the network is made of neurons, and that result may be outputted by some linear combination of the neurons in the last layer.}.
\end{remark}
}

\bibliographystyle{plainurl}

\appendix

%

\newpage

\section{Concepts from computable analysis}
\label{sec:defanalysecalculable}

When we say that a function $f: S_{1}  \times \dots \times S_{d} \to \R^{d'}$ is (respectively: polynomial-time) computable this will always be in the sense of computable analysis: see e.g. \cite{brattka2008tutorial,Wei00}. We recall here the basic concepts and definitions, mostly following the book \cite{Ko91}, whose subject is complexity theory in computable analysis. This section is repeating the formalisation proposed in \cite{BlancBournezMCU22vantardise} done to mix complexity issues dealing with integer and real arguments: 
a dyadic  number $d$ is a rational number with a finite binary expansion. That is to say $d=m / 2^{n}$ for some integers $m \in \Z$, $n\in \N$, $n \geq 0$. Let $\dyadic$ be the set of all dyadic rational numbers. We denote by $\dyadic_{n}$ the set of all dyadic rationals $d$ with a representation $s$ of precision $\operatorname{prec}(s)=n$; that is, $\dyadic_{n}=\left\{m \cdot 2^{-n} \mid m \in \Z\right\}$.

\begin{definition}[\cite{Ko91}]  \label{def:cinq} For each real number $x$, a function $\phi: \N \rightarrow \dyadic$ is said to binary converge to $x$ if for all $n \in \N, \operatorname{prec}(\phi(n))=n$ and $|\phi(n)-x| \leq 2^{-n}$. Let $C F_{x}$ (Cauchy function) denotes the set of all functions binary converging to $x$.
\end{definition}

Intuitively, a Turing machine $M$ computes a real function $f$ the following way: 1. The input $x$ to $f$, represented by some $\phi \in C F_{x}$, is given to $M$ as an oracle; 2. The output precision $2^{-n}$ is given in the form of integer $n$  as the input to $M$; 3. The computation of $M$ usually takes two steps, though sometimes these two steps may be repeated an indefinite number of times;
4. $M$ computes, from the output precision $2^{-n}$, the required input precision $2^{-m}$; 5. $M$ queries the oracle to get $\phi(m)$, such that $\|\phi(m)-x\| \leq 2^{-m}$, and computes from $\phi(m)$ an output $d \in \dyadic$ with $\|d-f(x)\| \leq$ $2^{-n}$.
%
More formally:

\begin{definition}[\cite{Ko91}] A real function $f: \R \rightarrow \R$ is computable if there is a function-oracle {TM} $M$ such that for each $x \in \R$ and each $\phi \in C F_{x}$, the function $\psi$ computed by $M$ with oracle $\phi$ (i.e., $\left.\psi(n)=M^{\phi}(n)\right)$ is in $C F_{f(x)}$. 
\shortermcu{We say the function $f$ is computable on the interval $[a, b]$ if the above condition holds for all $x \in[a, b]$.}
\end{definition}

\shortermcu{
\begin{remark}
Given some $x \in \R$, such  an oracle TM $M$ can determine some integer $X$ such that $x \in [-2^{X},2^{X}]$.
\end{remark}
}

\olivierplusimportant{Reste de MCU: Probablement superflus pour l'instant:
The following concept plays a very important role:

\begin{definition} \label{def:above}
Let $f:[a, b] \rightarrow \R$ be a continuous function on $[a, b]$. Then, a function $m: \N \rightarrow \N$ is said to be a modulus function of $f$ on $[a, b]$ if for all $n \in \N$ and all $x, y \in[a, b]$, we have
$$
|x-y| \leq 2^{-m(n)} \Rightarrow|f(x)-f(y)| \leq 2^{-n}
$$
\end{definition}

The following is well known (see e.g. \cite{Ko91} for a proof):

\begin{theorem}
A function $f: \R \rightarrow \R$ is computable iff there exist two recursive functions $m: \N \times \N \rightarrow \N$ and $\psi: \dyadic \times \N \rightarrow \dyadic$ such that
\begin{enumerate}
\item for all $k, n \in \N$ and all $x, y \in[-k, k],|x-y| \leq 2^{-m(k, n)}$ implies $|f(x)-f(y)| \leq 2^{-n}$, and
\item  for all $d \in \dyadic$ and all $n \in \N,|\psi(d, n)-f(d)| \leq 2^{-n}$.
\end{enumerate}
\end{theorem}
}

\shortermcu{
\subsection{On computable analysis: Complexity}
}

Assume that $M$ is an oracle machine computing $f$ on a do$\operatorname{main} G$. For any oracle $\phi \in C F_{x}$, with $x \in G$, let $T_{M}(\phi, n)$ be the number of steps for $M$ to halt on input $n$ with oracle $\phi$, and $T_{M}^{\prime}(x, n)=\max \left\{T_{M}(\phi, n) \mid \phi \in C F_{x}\right\}$. The time complexity of $f$ is defined as follows:

\begin{definition}[\cite{Ko91}]
Let $G$ be  bounded closed interval $[a, b]$. Let $f: G \rightarrow \R$ be a computable function. Then, we say that the time complexity of $f$ on $G$ is bounded by a function $t: G \times \N \rightarrow \N$ if there exists an oracle TM $M$ which computes $f$ 
such that for all $x \in G$ and all $n>0$, $T_{M}^{\prime}(x, n) \leq t(x, n)$.
\end{definition}

In other words, the idea is to measure the time complexity of a real function based on two parameters: input real number $x$ and output precision $2^{-n}$. Sometimes, it  is more convenient to simplify the complexity measure to be based on only one parameter, the output precision. For this purpose, we say the uniform time complexity of $f$ on $G$ is bounded by a function $t^{\prime}: \N \rightarrow \N$ if the time complexity of $f$ on $G$ is bounded by a function $t: G \times \N \rightarrow \N$ with the property that for all $x \in G$, $t(x, n) \leq t^{\prime}(n)$.

{
However, if we do so, it is important to realise that if we had taken $G=\R$ in the previous definition, for unbounded functions $f$, the uniform time complexity would not have existed, because the number of moves required to write down the integral part of $f(x)$ grows as $x$ approaches $+\infty$ or $-\infty$. Therefore, the approach of \cite{Ko91} is to do as follows (The bounds $-2^{X}$ and $2^{X}$ are somewhat arbitrary, but are  chosen here  because the binary expansion of any $x \in\left(-2^{n}, 2^{n}\right)$ has at most $n$ bits in the integral part).
%
%
\begin{definition}[Adapted from \cite{Ko91}]  For functions $f(x)$ whose domain is $\R$, 
 we say that the (non-uniform) time complexity of $f$ is bounded by a function $t^{\prime}: \N^{2} \rightarrow \N$ if the time complexity of $f$ on $\left[-2^{X}, 2^{X}\right]$ is bounded by a function $t: \N^{2} \rightarrow \N$ such that $t(x, n) \leq t^{\prime}(X, n)$ for all $x \in\left[-2^{X}, 2^{X}\right]$. 
\end{definition}

\begin{remark} 
The space complexity of a real function is defined similarly. We say the space complexity of $f: G \rightarrow \R$ is bounded by a function $s: G \times \N \rightarrow \N$ if there is an oracle TM $M$ which computes $f$ such that for any input $n$ and any oracle $\phi \in C F_{x}, M^{\phi}(n)$ uses $\leq s(x, n)$ cells, and the uniform space complexity of $f$ is bounded by $s^{\prime}: \N \rightarrow \N$ if for all $x \in G$ and all $\phi \in C F_{x}, M^{\phi}(n)$ uses $\leq s^{\prime}(n)$ cells. All coming definitions can then be easily extended to talk about space complexity.
\end{remark}

%

As we want to talk about general functions in $\Lesfonctionsquoi$, we extend the approach to more general functions.  
%
%
(for conciseness, when $\tu x=(x_{1},\dots,x_{p})$, $\tu X= (X_{1},\dots, X_{p})$, we write
$\tu x \in [-2^{\tu X}, 2^{\tu X}]$ as a shortcut for $x_{1} \in\left[-2^{X_{1}}, 2^{X_{1}}\right]$,  \dots, $x_{p} \in\left[-2^{X_{p}}, 2^{X_{p}}\right]$). 


\begin{definition}[Complexity for real functions: general case]   \label{def:bonendroit} Consider a function  $f(x_{1},\dots,x_{p}$, $n_{1},\dots,n_{q})$ whose domain is $\R^{p} \times \N^{q}$. 
 We say that the (non-uniform) time complexity of $f$ is bounded by a function $t^{\prime}: \N^{p+q+1} \rightarrow \N$ if the time complexity of $f(\cdot,\dots,\cdot,\ell(n_{1}),\dots,\ell(n_{q}))$ on $\left[-2^{X_{1}}, 2^{X_{1}}\right] \times \dots \left[-2^{X_{p}}, 2^{X_{p}}\right] $  
 is bounded by a function $t(\cdot,\dots,\cdot,\ell(n_{1}),\dots,\ell(n_{q}),\cdot): \N^{p} \times \N \to \N$ such that 
$ t(\tu x,\ell(n_{1}),\dots,\ell(n_{q}), n) \leq t^{\prime}(\tu X,\ell(n_{1}), \dots,\ell(n_{q}), n)$
 whenever  $\tu x \in \left[-2^{\tu X}, 2^{\tu X}\right].$
 We say that $f$ is polynomial time computable if $t^{\prime}$ can be chosen as a polynomial. 
 We say that a vectorial function is polynomial time computable iff all its components are. 
 \end{definition}
 
 \shortermcu{
 \begin{remark}
There is some important {subtlety}: when considering $f: \N \to \Q$, as $\Q \subset \R$, stating $f$ is computable may mean two things: in the classical sense, given integer $y$,  i.e. one can compute $p_y$ and $q_{y}$ some integers such that $f(y)=p_{y}/q_{y}$, or that it is computable in the sense of computable analysis: given some precision $n$,  given arbitrary $y$, and $n$ we can provide some rational (or even dyadic) $q_{n}$ such that $|q_{n}-f(y)| \leq 2^{-n}$. As we said, we always consider the latter.
\end{remark}
}

 We do that so this measure of complexity extends the usual complexity measure for functions over the integers, where complexity of integers is measured with respects of their lengths, and over the reals, where complexity is measured with respect to their approximation.
 %
 In particular, in the specific case of a function $f: \N^{d} \to \R^{d'}$, that basically means there is some polynomial $t': \N^{d+1} \to \N$ so that the time complexity of producing some dyadic approximating $f(\tu m)$ at precision $2^{-n}$ is bounded by $t'(\ell(m_{1}),\dots,\ell(m_{d}),n)$. 


%

In other words, when considering that a function is polynomial time computable, it is in the length of all its integer arguments, as this is the usual convention. 

\olivierplusimportant{Commenté, car pas besoin, je pense (enfin j'espère)
A well-known observation is the following.

\begin{theorem} Consider $\tu f$ as in Definition \eqref{def:bonendroit} computable in polynomial time. Then $\tu f$ has a polynomial modulus function of continuity, that is to say there is a polynomial function $m_{\tu f}: \N^{p+q+1}\rightarrow \N$ such that for all $\tu x,\tu y$ and all $n>0$, $\|\tu x-\tu y\| \leq 2^{-m_{\tu f}(\tu X,\ell(n_{1}),\dots,\ell(n_{q}), n)}$ implies 
$\|\tu f(\tu x,n_{1}, \dots,n_{q})-\tu f(\tu y,n_{1}, \dots,n_{q})\| \leq 2^{-n}$,
whenever $\tu x,\tu y  \in\left[-2^{\tu X}, 2^{\tu X}\right]$.
\end{theorem}
}

\section{Some results from $\THEPAPIERSPLUS$}
\label{sec:dode}

\subsection{Some general statements from \THEPAPIERS}

To be self-content, we recall in this section some results and concepts from \THEPAPIERS. All the statements in this section are already present in \THEPAPIERS: we are just repeating them here in case this helps.  We provide some of the proofs when they are not in the preliminary ArXiv version.  

As said in the introduction:

\begin{definition}[Discrete Derivative] The discrete derivative of
	$\tu f(x)$ is defined as $\Delta \tu f(x)= \tu f(x+1)-\tu
        f(x)$. We will also write $\tu f^\prime$ for
        $\Delta \tu f(x)$ to help readers not familiar with discrete differences to understand
	statements with respect to their classical continuous counterparts. 
	
\end{definition}

Several results from classical derivatives generalise
to the settings of discrete differences: this includes linearity of derivation $(a \cdot
f(x)+ b \cdot g(x))^\prime = a \cdot f^\prime(x) + b \cdot
g^\prime(x)$, formulas for products
and division such as
 $(f(x)\cdot g(x))^\prime =
 f^\prime(x)\cdot g(x+1)+f(x) \cdot g^\prime(x)= f(x+1)  g^\prime(x) +  f^\prime(x)  g(x)$. 
    Notice that, however, there is no simple equivalent of the chain rule, in other words, there is no simple formula for the derivative of the composition of two functions. 
 
A fundamental concept is the following:

\begin{definition}[Discrete Integral]
	Given some function $\tu f(x)$, we write $$\dint{a}{b}{\tu f(x)}{x}$$
	as a synonym for $\dint{a}{b}{\tu f(x)}{x}=\sum_{x=a}^{x=b-1}
        \tu f(x)$ with the convention that it takes value $0$ when $a=b$ and
        $\dint{a}{b}{\tu f(x)}{x}=- \dint{b}{a}{\tu f(x)}{x}$ when $a>b$. 
\end{definition}

The telescope formula yields the so-called Fundamental Theorem of
Finite Calculus: 

\begin{theorem}[Fundamental Theorem of Finite Calculus]
	Let $\tu F(x)$ be some function.
	Then,
	$$\dint{a}{b}{\tu F^\prime(x)}{x}= \tu F(b)-\tu F(a).$$
      \end{theorem}

A classical concept in discrete calculus is the one of falling
	power defined as $$x^{\underline{m}}=x\cdot (x-1)\cdot (x-2)\cdots(x-(m-1)).$$
	This notion is  motivated by the fact that it satisfies a derivative formula $
        (x^{\underline{m}})^\prime  = m \cdot x^{\underline{m-1}}$  similar to the classical
        one for powers in the continuous setting.
In a similar spirit, we 
introduce the concept of falling exponential. 


\begin{definition}[Falling exponential]
	Given some function $\tu U(x)$, the expression $\tu U$ to the
	falling exponential $x$,
	denoted by $\fallingexp{\tu U(x)}$, stands
        for  \begin{eqnarray*}
                \fallingexp{\tu U(x)} &=&
                                                                                (1+ \tu U^\prime(x-1)) \cdots
                                        (1+ \tu U^\prime(1)) \cdot (1+ \tu U^\prime(0))   \\
                                         &=&
                   \prod_{t=0}^{t=x-1} (1+ \tu U^\prime(t)),
                    \end{eqnarray*}
	with the convention that $\prod_{0}^{0}=\prod_{0}^{-1}=\tu {id}$, where $\tu
        {id}$ is the identity (sometimes denoted  $1$ hereafter).
\end{definition}

This is motivated by the remarks that 
	$2^x=\fallingexp{x}$, and
        that the discrete
	derivative of a falling exponential is given by
	$$\left(\fallingexp{\tu U(x)}\right )^\prime = \tu U^\prime(x) \cdot
	\fallingexp{\tu U(x)}$$
	%
      for all $x \in \N$.

\begin{lemma}[Derivation of an integral with parameters]  \label{derivationintegral}
   Consider $$\tu F(x) = \dint{a(x)}{b(x)} {\tu f(x,t)}{t}.$$
   Then \begin{eqnarray*}
          \tu F'(x) &=& \dint{a(x)}{b(x)}{  \frac{\partial \tu f}{\partial
                        x} (x,t)}{t} 
                     + \dint{0}{-a^\prime(x)}{\tu f(x+1,a(x+1)+t)}{t} 
+ \dint{0}{b'(x)}{ \tu f(x+1,b(x)+t ) } {t}. 
\end{eqnarray*}

\noindent In particular, when $a(x)=a$ and $b(x)=b$ are constant functions, $\tu F'(x) = \dint{a}{b}{  \frac{\partial \tu f}{\partial
     x} (x,t)}{t},$
     and when $a(x)=a$ and $b(x)=x$,
     $\tu F'(x) = \dint{a}{x}{  \frac{\partial \tu f}{\partial
     x} (x,t)}{t} + \tu f(x+1,x)$.

\end{lemma}

 \begin{proof}
\begin{eqnarray*}
  \tu F'(x) &=& \tu F(x+1) - \tu F(x) \\
  &=& \sum_{t= a(x+1)}^{b(x+1) - 1} \tu f(x+1,t) -
                          \sum_{t= a(x)}^{b(x) - 1} \tu f(x,t)   \\
  &=& \sum_{t= a(x)}^{b(x)-1} \left( \tu  f(x+1,t) - \tu f(x,t)
      \right)    
     +   \sum_{t=a(x+1)} ^{t= a(x)-1} \tu f(x+1,t) 
                      + \sum_{t=b(x)}
      ^{b(x+1)-1} \tu f(x+1,t) \\
  &=& \sum_{t= a(x)}^{b(x)-1}  \frac{\partial \tu f}{\partial
      x} (x,t)  + \sum_{t=a(x+1)} ^{t= a(x)-1} \tu f(x+1,t) + \sum_{t=b(x)}
      ^{b(x+1)-1} \tu f(x+1,t) \\
  &=& \sum_{t= a(x)}^{b(x)-1}   \frac{\partial \tu f}{\partial
      x} (x,t)  + \sum^{t=-a(x+1)+a(x)-1} _{t=0} \tu f(x+1,a(x+1)+t) \\
      &&
  + \sum_{t=0}
      ^{b(x+1)-b(x)-1} \tu  f(x+1,b(x)+t).
\end{eqnarray*}
\end{proof}

\begin{lemma}[Solution of linear ODE
	] \label{def:solutionexplicitedeuxvariables}
	For matrices $\tu A$ and vectors $\tu B$ and $\tu G$,
	the solution of equation $\tu f^\prime(x,\tu y)= \tu A(\tu f(x,\tu y),\tu h(x, \tu y), x,\tu y) \cdot \tu f(x,\tu y)
	+  \tu
	B (\tu f(x,\tu y), \tu h(x, \tu y),  x,\tu y)$  with initial conditions $\tu f(0,\tu y)= \tu G(\tu y)$ is
	\begin{eqnarray*}\label{soluce}
        \tu f(x,\tu y)  &=&
          \left( \fallingexp{\dint{0}{x}{\tu
                             A(\tu f(t,\tu y),\tu h(t, \tu y), t,\tu y)}{t}} \right) \cdot \tu G (\tu
                             y)  \\
           &&
          +
	\dint{0}{x}{ \left(
		\fallingexp{\dint{u+1}{x}{\tu A(\tu f(t,\tu y),\tu h(t, \tu y), t,\tu y)}{t}} \right) \cdot
              \tu B(\tu f(u,\tu y),\tu h(u, \tu y), u,\tu y)} {u}.
          \end{eqnarray*}

        \end{lemma}
                
\begin{proof}
Denoting the right-hand side by $\tu {rhs}(x,\tu y)$, we  have 
$$\begin{array}{ccl} \bar {\tu {rhs}}^{\prime}(x,\tu y)&=& \tu A(\tu f(x,\tu y),\tu h(x, \tu y), x,\tu y)  \cdot  \left( \fallingexp{\dint{0}{x}{\tu
                             A(\tu f(t,\tu y),\tu h(t, \tu y),  t,\tu y)}{t}} \right)   \cdot \tu G (\tu
                             y) \\
                                     && +
                             \dint{0}{x}{ \left(
		\fallingexp{\dint{u+1}{x}{\tu A(\tu f(t,\tu y), \tu h(t, \tu y),  t,\tu y)}{t}} \right)^{\prime} \cdot
              \tu B(\tu f(u,\tu y), \tu h(u, \tu y), u,\tu y)} {u} \\ 
              && +
                             \left(
		\fallingexp{\dint{x+1}{x+1}{\tu A(\tu f(t,\tu y),\tu h(t, \tu y), t,\tu y)}{t}} \right) \cdot
              \tu B(\tu f(x,\tu y), \tu h(x, \tu y), x,\tu y) \\
              &=& \tu A(\tu f(x,\tu y),\tu h(x, \tu y), x,\tu y)  \cdot  \left( \fallingexp{\dint{0}{x}{\tu
                             A(\tu f(t,\tu y),\tu h(t, \tu y),  t,\tu y)}{t}} \right)   \cdot \tu G (\tu
                             y)  \\ && + \  \tu A(\tu f(x,\tu y), \tu h(x, \tu y), x,\tu y) \cdot \\
                             && \dint{0}{x}{ \left(
		\fallingexp{\dint{u+1}{x}{\tu A(\tu f(t,\tu y),\tu h(t, \tu y), t,\tu y)}{t}} \right)  \tu B(\tu f(u,\tu y), \tu h(u, \tu y), u,\tu y)} {u} \\
                    &&
              + \ \tu B(\tu f(x,\tu y),\tu h(x, \tu y), x,\tu y)
              \\
                            &=&  \tu A(\tu f(x,\tu y), \tu h(x, \tu y), x,\tu y) \cdot \tu {rhs}(x,\tu y) 	+  \tu
	B (\tu f(x,\tu y),\tu h(x, \tu y), x,\tu y)
	              \end{array} 
              $$
              where we have used linearity of derivation and definition of falling exponential for the first term, and derivation of an integral (Lemma \ref{derivationintegral}) providing the other terms to get the first equality, and then the definition of falling exponential.
              This proves the property by unicity of solutions of a discrete ODE, observing that $\bar {\tu {rhs}}(0,\tu y)=\tu G(\tu y)$.
\end{proof}

We write also $1$ for the identity. 

        \begin{remark} \label{rq:fund}
          Notice that this can  be rewritten  as 
          \begin{equation} \label{eq:rq:fund} 
\tu f(x,\tu y)=\sum_{u=-1}^{x-1}  \left(
\prod_{t=u+1}^{x-1} (1+\tu A(\tu f(t,\tu y), \tu h(t, \tu y),  t,\tu y)) \right) \cdot  \tu B(\tu f(u,\tu y), \tu h(u, \tu y), u,\tu y)
,
\end{equation}
with the (not so usual) conventions that for any function $\kappa(\cdot)$,  $\prod_{x}^{x-1} \tu \kappa(x) = 1$ and $\tu
B(-1,\tu y)=\tu G(\tu y)$.
Such equivalent expressions both have a clear computational content. They can
be interpreted as an algorithm unrolling
the computation of   $\tu f(x+1,\tu y)$ from the computation of  $\tu
f(x,\tu y), \tu f(x-1,\tu y), \ldots, \tu f(0,\tu y)$.  
        \end{remark}
        
A fundamental fact is that the derivation with respects to length provides a way to so a kind of change of variables. This is Lemma \ref{myfundobge} in the body of this article. 

\newcommand\PREUVESOIXANTCINQ{

\section{Some statements from \cite{BlancBournezMCU22vantardise}}
\label{ouestlapreuve}

We repeat here some statements from  \cite{BlancBournezMCU22vantardise}. They are motivated by repeating the arguments for proving
Proposition  \ref{prop:mcu:un}, 1.

\subsection{On the complexity of solving a linear length ODE}

We have to prove that all functions of $\manonclass$ are computable (in the sense of computable analysis) in polynomial time. 
The hardest  part is to prove that the class of polynomial time computable functions is preserved by the linear length ODE schema, so we start by it.


\begin{lemma}\label{lem:un}
	The class of polynomial time computable functions is preserved by the linear length ODE schema.
\end{lemma}

\olivierplusimportant{Avis au peuple: Notation $\tnorm{.}$ pourrie? euh, est-ce pas bien de noter comme ca un truc qui n'est pas classique avec une notation qui veut en général dire autre chose.} 
\manonplusimportant{pas d'accord: pour moi ça permet de bien la différentier des autres normes, et on la définit proprement}

We write $\tnorm{\cdots}$ for the sup norm of integer part: given some matrix $\tu
A=(A_{i,j})_{1 \le i \le n, 1 \le j \le m}$, 
$\tnorm{\tu A}=\max_{i,j}
\lceil A_{i,j} \rceil $. In particular, given a vector $\tu x$, it can be seen as a matrix with $m=1$, and $\tnorm{\tu x}$ is the sup norm of the integer part of its components.

Before proving this lemma, we prove the following result:


\begin{lemma}[Fundamental observation] \label{fundamencoreg}
	Consider the ODE 
	\begin{equation} \label{eq:bcg}
		\tu F^\prime(x,\tu y)=  {\tu A} ( \tu F(x,\tu y), \tu h({x},\tu y),
		{x},
		\tu y) \cdot
		\tu F(x,\tu y)
		+   {\tu B} ( \tu F(x,\tu y), \tu h({x},\tu y),
		{x},
		\tu y).
	\end{equation}
	Assume:
	\begin{itemize}
		\item  The initial condition $\tu G(\tu y) =^{def}
		\tu F(0, \tu y)$ is polynomial time computable, and  $\tu h({x},\tu y)$ 
		are polynomial time computable \unaire{x}. 
		\item   ${\tu A} ( \tu F(x,\tu y), \tu h ({x},\tu y),
		{x},
		\tu y)$ and ${\tu B} ( \tu F(x,\tu y), \tu h({x},\tu y),
		{x},
		\tu y)$ are \polynomialb{} expressions essentially constant in $\tu F(x,\tu y)$.
		
		%
	\end{itemize}
	Then, there exists a polynomial $p$ such that $\length{\tnorm{\tu F \left(x,\tu y \right) }}\leq p\left(x,\length{\tnorm{\tu y}}\right)$ and 
	$\tu F(x,\tu y)$ is polynomial time computable \unaire{x}.
\end{lemma}

\begin{proof}
%
%
%
%
%
%
The solution of ordinary differential equation \eqref{eq:bcg} can be put in some explicit form  (this follows from \THEPAPIERS, or alternatively Remark~\ref{rq:fund} following Lemma \ref{def:solutionexplicitedeuxvariables} in appendix, or, if you prefer, can be directly established by induction, using the recurrence formula \eqref{eq:bcg}).

\begin{equation} \label{formulemagique}
\tu F(x,\tu y)=\sum_{u=-1}^{x-1}  \left(
\prod_{t=u+1}^{x-1} (1+\tu A(\tu F(t,\tu y), \tu h(t, \tu y),  t,\tu y)) \right) \cdot  \tu B(\tu F(u,\tu y), \tu h(u, \tu y), u,\tu y)
.
\end{equation}
\noindent with the conventions that $\prod_{x}^{x-1} \tu \kappa(x) = 1$ and $\tu B( \cdot , -1,\tu y)=\tu G(\tu y)$. 		

This formula permits to evaluate $\tu F(x, \tu y)$, using a dynamic programming approach,  from the quantities $\tu y$, $u$, $\tu h(u,\tu y)$, for $1 \le u < x$,  in a number of arithmetic steps that is polynomial in $x$. Indeed: 
for any $-1 \le u \le x$, $\tu A(\tu F(u,\tu y),\tu h(u,\tu y), u,\tu y)$     and 
$\tu B( \tu F(u,\tu y), \tu h(u,\tu y), u,\tu y)$ are matrices whose coefficients are \polynomialb{}. Their coefficients involve finitely many arithmetic operations or $\signb{}$ operations from their inputs. Once this is done, computing $\tu F(x, \tu y)$  requires polynomially in $x$ many arithmetic operations: basically, once the values for $\tu A$ and $\tu B$ are known we have to sum up $x+1$ terms, each of them involving at most $x-1$ multiplications.

We need to take  care not only on the arithmetic complexity that is polynomial in $x$, but also on the bit complexity. We start by discussing the bit complexity of the integer parts. Each of the quantities $\tu y$, $u$, $\tu h(u,\tu y)$, for $1 \le u < x$, being polynomial time computable,   has its integer part with a bit complexity that remains polynomial in $x$ and $\length{\tnorm{y}}$.  As the bit complexity of a sum, product, etc is polynomial in the size of its arguments, it is sufficient to show  that the growth rate of function $\tu F(x,\tu y)$ can be polynomially dominated. For this, recall that, for any $-1 \le u \le x$, coefficients of $\tu A(\tu F(u,\tu y),\tu h(u,\tu y), u,\tu y)$     and 
$\tu B( \tu F(u,\tu y), \tu h(u,\tu y), u,\tu y)$ are essentially constant in $\tu F(u,\tu y)$. Hence, the size of the integer part of these coefficients do not depend on $\length{\tu F(u,\tu y)}$. Since, in addition, $\tu h$ is computable in polynomial time in $x$ and $\length{\tu y}$, there exists a polynomial $p_M$ such that:

\begin{equation} \label{taille}
\max (\length{\tnorm{\tu A(\tu F(u,\tu y),\tu h(u,\tu y), u,\tu y)}},\length{\tnorm{\tu B(\tu F(u,\tu y),\tu h(u,\tu y), u,\tu y)} \leq p_M(u, \length{\tnorm{\tu y}}}). 	
\end{equation}

It then holds that,
\[
	\length{\tnorm{\tu F(x+1, \tu y)}}\leq p_M(x,\length{\tnorm{\tu y}}) + \length{\tnorm{\tu F(x, \tu y)}} +2 	
\]

It follows from an easy induction that we must have
$$\length{\tnorm{\tu F(x, \tu y) }} \le \length{\tnorm{G(\tu y)}} + x \cdot (p_{M}(x, \length{ \tnorm{\tu y}}+2),$$ which gives the desired bound on the length of the integer part of the values for function $\tu F$, after observing that, since $\tu G$ is polynomial time computable, necessarily $\length{\tnorm{G(\tu y)}}$ remains polynomial in $\length{ \tnorm{\tu y}}$.

We  now  take care of the bit complexity of involved quantities, in order to prove that $\tu F(x,\tu y)$ is indeed polynomial time computable \unaire{x}. Given $\tu y$, we can determine some integer $Y$ such that $\tu y \in [2^{-Y},2^{Y}]$. 
We just need to prove that given $n$, we can provide some dyadic $\tu z_{x}$ approximating $\tu F(x,\tu y)$ at precision $n$, i.e. with $\| \tu F(x,\tu y) - \tu z_{x}\| \le 2^{-n}$, in a time polynomial in $x$ and $Y$. 

Basically, this follows from the possibility of evaluating $\tu F(x,\tu y)$ using formula \eqref{formulemagique}. This latter formula is made of a sum of $x+1$ terms, that we can call  $\tu T_{0}, \tu T_{1}, \dots, \tu T_{x}$, each of them $\tu T_{i}$ corresponding to a product of $k$ matrices (or vectors) $\tu T_{i}=\tu C_{1} \tu C_{2} \dots  \tu C_{k(i)}$, with $k(i) \le x+1$, were each $\tu C_{j}$ is either some $\tu B(\tu F(u,\tu y), \tu h(u, \tu y), u,\tu y)$ for some $u$ 
or some $(1+\tu A(\tu F(t,\tu y), \tu h(t, \tu y),  t,\tu y))$ for some $t$. 

To solve our problem, it is sufficient to be able to approximate the value of each $\tu T_{i}$ by some dyadic $\tu d_{i}$ with precision $2^{-n-m}$, considering $m$ with $x+1 \le 2^{m}$. Indeed, taking $\tu z_{x}=\sum_{i=0}^{x} \tu d_{i}$ will guarantee an error on the approximation of $\tu F(x,\tu y)$
 less than $(x+1) 2^{-n-m} \le 2^{-n}.$
 
  \newcommand\tC{\tilde{C}}
 So we focus on the problem of estimating  $\tu T_{i}=\tu C_{1} \tu C_{2} \dots  \tu C_{k(i)}$ with precision $2^{-n-m}$. If we write $\tu \tC_{j}
 $ for some approximation of $\tu C_{j}$, 
 we can write
 
\begin{equation} \label{erreurproduit}
 \begin{array}{lll}
 \left\| \tu C_{1} \tu C_{2} \dots  \tu C_{k(i)} - \tu \tC_{1} \tu \tC_{2} \dots  \tu \tC_{k(i)} \right\|
 		\le &&  \left\|  \tu C_{1} \tu C_{2} \dots  \tu C_{k(i)-1} \left(\tu C_{k(i)}-\tu \tC_{k(i)}\right)  \right\| \\
		&   +& \left\|  \tu C_{1} \tu C_{2} \dots  \tu C_{k(i)-2} \left(\tu C_{k(i)-1}-\tu \tC_{k(i)-1}\right) \tu \tC_{k(i)}  \right\| \\
		&   +& \left\|  \tu C_{1} \tu C_{2} \dots  \tu C_{k(i)-3} \left(\tu C_{k(i)-2}-\tu \tC_{k(i)-2}\right) \tu \tC_{k(i)-1}  \tu \tC_{k(i)}  \right\| \\
		& \vdots & \\ 
		& + & \left\|  \left(\tu C_{1}-\tu \tC_{1}\right) \tu \tC_{1}  \tu \tC_{2}  \dots \tu \tC_{k(i)}  \right\|. \\
\end{array}
\end{equation}

We just need then to compute some approximation $ \tu \tC_{k(i)}$ of $\tu C_{k(i)}$, guaranteeing the first term in \eqref{erreurproduit} to be less than $2^{-n-m-m}$,
and then choose some approximation $ \tu \tC_{k(i)-1}$ of $\tu C_{k(i)-1}$, guaranteeing the second term in \eqref{erreurproduit} to be less than $2^{-n-m-m}$, and so on.  Doing so, we will 
solve our problem, as $\tu \tC_{1} \tu \tC_{2} \dots  \tu \tC_{k(i)}$ will provide an estimation of $\tu T_{i}$, with an error  less than  $(x+1) 2^{-n-m-m} \le 2^{-n-m}$ by
\eqref{erreurproduit}.

For the first term of \eqref{erreurproduit}, the point is that we know from a reasoning similar to previous computations (namely bounds such as \eqref{taille}) that we have
$$\tnorm{\tu C_{1} \tu C_{2} \dots  \tu C_{k(i)-1}} \le 2^{p_{1}(x,Y)}$$
for some polynomial $p_{1}$.

Consequently, it is sufficient to take $\left\| \tu C_{k(i)}-\tu \tC_{k(i)} \right\| \le 2^{-n-2m-p_{1}(x,Y)}$ to guarantee an error less than $2^{-n-m-m}$.

A similar analysis applies for the second, and other terms, that involve finitely many terms. The whole approach takes a time that remains clearly polynomial in $x$ and $Y$. 
\end{proof}

The previous statements lead to the following:

\begin{lemma}[Intrinsic complexity of linear $\lengt$-ODE]~\label{lem:fundamentalobservationlinearlengthODE}
	Let $\tu f$ be a solution of the linear $\lengt$-ODE 
	\begin{eqnarray*}
	\tu f(0,\tu y) &=& \tu g(\tu y), \\
	\dderivl{\tu f(x,\tu y)}&=&   \tu u(\tu f(x,\tu y), \tu h(x,\tu y),
	x,\tu y) 
	\end{eqnarray*}
	\noindent where $\tu u$ is \textit{essentially linear} in $\tu f(x, \tu y)$. 
	Assume that the functions $\tu u, \tu g, \tu h$ are computable in polynomial time. 
	Then, $\tu f$ 
	is computable in polynomial time.
\end{lemma}

\begin{proof}
From Lemma~\ref{myfundobge}, $\tu f(x,\tu y)$ can also be given by 
$\tu f(x,\tu y)= \tu F(\length{x},\tu y)$
where $\tu {\bar F}$ is the solution of the initial value problem 
\begin{eqnarray*}
\tu {\bar F}(1,\tu y)&=& \tu g(\tu y), \\
\dderiv{\tu {\bar F}(t,\tu y)}{t} &=& \tu u(\tu {\bar F}(t, \tu y),2^{t}-1,\tu y).
\end{eqnarray*}

	Functions $\tu u$ are \polynomialb{} expressions that are essentially linear in $\tu f(x,\tu y)$. So there exist matrices $\tu A$, $\tu B$ that are essentially constants in $\tu f(t,\tu y)$ such that 
	
	\[
		\dderivl{\tu f(x,\tu y)} = {\tu A} ( \tu f(x,\tu y), \tu h(x,\tu y),
		x,\tu y) \cdot
		  \tu f(x,\tu y)
		  +   {\tu B} ( \tu f(x,\tu y), \tu h(x,\tu y),
		  x,\tu y).		
	\]

	\noindent In other words, it holds
	$$
	\tu F’(t,\tu y)= \overline {\tu A} ( \tu F(t,\tu y),
	t,\tu y) \cdot
	\tu F(t,\tu y)
	+   \overline {\tu B} ( \tu F(t,\tu y),
	t,\tu y).
	$$
	by setting 
	\begin{eqnarray*}
		\overline {\tu A} ( \tu F(t,\tu y),
		t,\tu y) &=&  {\tu A} (\tu { \bar F}(t, \tu y),\tu h(2^{t}-1,\tu y), 2^{t}-1,\tu y) \\
		\overline {\tu B} ( \tu F(t,\tu y),
		t,\tu y) &=& {\tu B} (\tu { \bar F}(t, \tu y),\tu h(2^{t}-1,\tu y), 2^{t}-1,\tu y)	
	\end{eqnarray*}

	The
	corresponding matrix $\overline {\tu A}$ and vector $\overline {\tu
		B}$ are essentially constant in $\tu F(t, \tu y)$. Also, 
	functions $\tu g, \tu h$ are computable in polynomial time, more precisely 
	polynomial in  $\length{x}$, hence in $t$, and $\length{y}$. Given $t$, obtaining $2^{t}-1$ is immediate. This guarantees that all
	hypotheses of Lemma \ref{fundamencoreg} are true. We can then conclude remarking, again, that $t=\length{x}$. 	
\end{proof}

This proves Lemma \ref{lem:un}.

\subsection{Proof of  Proposition \ref{prop:mcu:un}, 1.}

We can now state and prove Proposition \ref{prop:mcu:un}, 1: 
\\


\begin{proof}[Proof of Proposition \ref{prop:mcu:un}, 1.]
	This is proved by induction.  This is true for basis functions, from basic arguments from computable analysis. In particular as $\signbname$ is a continuous piecewise affine function with rational coefficients, it is computable in polynomial time from standard arguments. 
	
	Now, the class of polynomial time computable functions is  preserved by composition. This is proved in \cite{Ko91}: in short, the idea of the proof for $COMP(f,g)$, is that by induction hypothesis, there exists $M_f$ and $M_g$ two Turing machines computing in polynomial time $f: \RR \rightarrow \RR$ and $g : \RR \rightarrow \RR$. In order to compute $COMP(f,g)(x)$ with precision $2^{-n}$, we just need to compute $g(x)$ with a precision $2^{-m(n)}$, where $m(n)$ is the polynomial modulus function of $f$. 
	%
	Then, we compute $f(g(x))$, which, by definition of $M_f$ takes a polynomial time in $n$. 
	Thus, since polynomial time with an oracle in polynomial time is polynomial time, $COMP(f,g)$ is computable in polynomial time, so the class of polynomial time computable functions is preserved under composition.
	It only remains to prove that the class of polynomial time computable functions is preserved by the linear length ODE schema: this is Lemma \ref{lem:un}. 
\end{proof}
}
\PREUVESOIXANTCINQ

%
%
%
%
%
%

%
%
%
%
%
%

%

\end{document}